\documentclass[a4paper, amsfonts, amssymb, amsmath, reprint, showkeys, nofootinbib, twoside, floatfix,superscriptaddress, aps, pra]{revtex4-1}
\usepackage[english]{babel}
\usepackage{braket}
\usepackage{dsfont}
\usepackage{subcaption}
\usepackage{siunitx}
\usepackage{amsthm}
\usepackage{mathtools}
\usepackage{physics}
\usepackage{xcolor}
\usepackage{graphicx}
\usepackage[left=23mm,right=13mm,top=35mm,columnsep=15pt]{geometry} 
\usepackage{adjustbox}
\usepackage{placeins}
\usepackage[T1]{fontenc}
\usepackage{lipsum}
\usepackage{csquotes}
\usepackage{amsmath,amssymb}
\usepackage{graphicx}
\usepackage{tabularx,booktabs}
\usepackage{tikz}
\usepackage{float}
\usetikzlibrary{quantikz2}
\usepackage[utf8]{inputenc}
\usepackage[colorinlistoftodos, color=green!40, prependcaption]{todonotes}
\usepackage[pdftex, pdftitle={Article}, pdfauthor={Author}]{hyperref}
\usepackage[ruled,lined]{algorithm2e}
\usepackage{algpseudocode}
\usepackage{hyperref}

\usepackage{caption}

\newtheorem{lemma}{Lemma}
\newtheorem{corollary}{Corollary}

\theoremstyle{definition}

\usetikzlibrary{shapes.geometric, arrows}
\usetikzlibrary{automata,positioning}

\tikzstyle{startstop} = [rectangle, rounded corners, minimum width=3cm, minimum height=1cm, text centered, draw=black, fill=blue!20]
\tikzstyle{io} = [trapezium, trapezium left angle=70, trapezium right angle=110, minimum width=3cm, minimum height=1cm, text centered, draw=black, fill=blue!30]
\tikzstyle{process} = [rectangle, minimum width=3cm, minimum height=1cm, text centered, draw=black, fill=blue!20]
\tikzstyle{decision} = [diamond, minimum width=3cm, minimum height=1cm, text centered, draw=black, fill= green!30]
\tikzstyle{arrow} = [thick, ->, >=stealth]

\begin{document}
\title{A quantum double-bracket algorithm for imaginary-time evolution with exponentially shorter depth}
\author{Ioannis Kolotouros}
    \email{ioannis@phasecraft.io}
    \affiliation{Phasecraft Ltd., London, United Kingdom}

\author{Raul Garcia-Patron}
    \email{raulgp@phasecraft.io}
    \affiliation{Phasecraft Ltd., London, United Kingdom}

\date{\today}

\begin{abstract}
Imaginary-time evolution (ITE) is a widely used technique for preparing the ground state of a target Hamiltonian. Building on the recently proposed framework of Double-Bracket quantum ITE (DB-QITE), we introduce a probabilistic variant (PDBQITE) that achieves an exponential reduction in circuit depth. Our algorithm requires only a single ancilla qubit and a quantum circuit that consists of one controlled real-time evolution and one mid-circuit measurement per iteration. The trade-off is an exponential decay of the total success probability with the number of iterations, for which we derive an analytical lower bound that is independent of system size. Our numerical simulations show that the per-step success probability remains close to unity, making it realistic to execute many iterations, thus providing a substantial improvement over prior constructions. We show that PDBQITE outperforms the near-term probabilistic algorithm PITE for ground-state preparation of molecular Hamiltonians, and that a hybrid QAOA–PDBQITE scheme achieves higher approximation ratios than QAOA at the same circuit depth on random regular graphs for the MaxCut problem.

\end{abstract}

\maketitle

\section{Introduction}
Quantum computing is expected to deliver speedups in tasks ranging from optimization \cite{vcepaite2025quantum, boulebnane2024solving, montanaro2024quantum, he2026regularized} and cryptography \cite{shor1999polynomial, babbush2026securing}, to material science \cite{alam2025programmable, alam2025fermionic} and chemistry \cite{zhang2025quantum, low2025fast}. In recent years, advances on the hardware side have progressed at a high pace, with even small demonstrations of quantum error-correction \cite{google2025quantum}. However, there is still great need to understand the utility of these near-term (and small-scale) devices and whether they can offer advantages over their classical counterparts.

A prominent example of an algorithm that can run on near-term devices is imaginary-time evolution (ITE), which allows the preparation of ground states or Gibbs states. In ITE, the quantum state is evolved with the use of the non-unitary operator $e^{-H\tau}$, with $\tau\geq 0$ and $H$ being the problem Hamiltonian. This evolution is guaranteed to converge to the ground state at $\tau \to \infty$ as long as the initial state of the system has a non-zero overlap with the ground state. However, implementing this non-unitary operator on a quantum computer is a nontrivial task. 

Over the past few years, several different algorithms have been developed to run ITE on a quantum device. For example, in \cite{motta2020determining}, the authors developed a tomographic approach, which they name quantum ITE (QITE), in which they find a unitary operator solving a linear system, that allows them to approximate the normalized imaginary-time evolved state at the cost of an exponential scaling with the correlation length. Furthermore, other approaches use a parameterized quantum circuit, and employ McLachlan’s variational principle to update the parameters so that the generated quantum states follow the ITE trajectory \cite{mcardle2019variational}. However, these methods can either exhibit exponentially small gradients if the initial state is not chosen carefully, or they can incur impractically high costs \cite{gacon2024variational}, as they require the estimation of the quantum Fisher information matrix \cite{liu2020quantum}. Finally, there exist probabilistic methods \cite{kosugi2022imaginary, nishi2023optimal, nishi2025encoded, zhang2025quantum, xie2024probabilistic}, in which the non-unitary operator is encoded as a block of a larger unitary operator. However, due to the probabilistic nature of these methods, performing a large number of iterations can lead to a probability of success that decays exponentially.

Very recently, Gluza et al. \cite{gluza2026double} introduced a novel double-bracket QITE (DB-QITE) algorithm, in which unitary operators are constructed recursively so that they approximate the (normalized) ITE states. Their algorithm is motivated by the fact that ITE states are solutions to a specific class of differential equations known as double-bracket flows \cite{brockett1991dynamical}. Although their algorithm exhibits important cooling properties and can be executed \emph{coherently} in a quantum computer, its biggest bottleneck is that it requires quantum circuits with depth that grows exponentially with the number of ITE steps. Hence, their algorithm becomes impractical already when implementing a few steps on a near-term quantum device. 

In this paper, we resolve this depth bottleneck by introducing a probabilistic double-bracket QITE algorithm (PDBQITE) whose circuit depth grows linearly with the number of iterations, in exchange for a per-iteration success probability. Our algorithm turns each recursion step of Gluza et al. into a linear combination of unitaries (LCU) \cite{childs2012hamiltonian}, implemented with a single ancilla qubit and a single controlled real-time evolution of the target Hamiltonian $H$ per step. Crucially, we prove an analytical lower bound on the per-iteration success probability that is independent of the system size and close to unity. Hence, despite suffering from an exponential decay of success probability, our algorithm allows for a much larger number of iterations compared to prior approaches.

Furthermore, we show how to extend our algorithm to a second-order PDBQITE that allows for larger timesteps per iteration, at the cost of an additional ancilla qubit and an extra controlled real-time evolution; an approach that would have been computationally prohibitive following the original formulation of Gluza et al. \cite{gluza2026double}. We also show how one can achieve modest gains in the total probability of success by employing an amplitude amplification strategy inspired by \emph{robust oblivious amplitude amplification} \cite{berry2015simulating, berry2015hamiltonian} and \emph{fixed-point amplitude amplification} \cite{yoder2014fixed}, but at the cost of slower convergence.

We validate PDBQITE with classical simulations on two problem families. For ground-state preparation of the molecules $\mathrm{LiH}$ and $\mathrm{BeH_2}$, starting from the Hartree–Fock state, both first- and second-order PDBQITE reach chemical accuracy; the second-order variant does so in as few as two iterations, and although its per-step success probability is lower, its cumulative success probability at chemical accuracy is roughly an order of magnitude higher than that of first-order. In the same setting, PDBQITE outperforms the near-term probabilistic method PITE, reaching lower energy per iteration with a higher success probability. For combinatorial optimization, we study a hybrid scheme in which a QAOA circuit is used as a starting point for PDBQITE. We find that on randomly weighted 3- and 4-regular MaxCut instances of up to $n=18$ qubits, appending a few PDBQITE steps yields higher approximation ratios than spending the same circuit depth on additional QAOA layers, driving the ratio close to one. Across all instances, the measured per-step success probability closely tracks our system-size-independent lower bound.

\emph{Structure:} In Sec. \ref{sec:preliminaries} we give a basic introduction on ITE and on the recently-proposed DB-QITE algorithm. In Sec. \ref{sec:pdbqite} we give an overview of our algorithm and explain how to implement this on a quantum device. In Sec. \ref{sec:experiments}, we evaluate how our algorithm works in practice by testing its performance on the MaxCut problem and on the task of preparing ground states of molecular Hamiltonians. In Sec. \ref{sec:main_results}, we present the two main versions of our algorithm. Finally, in Sec. \ref{sec:discussion} we conclude with a small discussion about our algorithm and future directions.

\section{Preliminaries}
\label{sec:preliminaries}

\subsection{Imaginary-time Evolution}
\label{sec:imaginary_time_evolution}


Consider a state $\ket{\psi(0)}$ whose interactions are described by a time-independent Hamiltonian $H$. Its dynamics are governed by the time-dependent Schrödinger equation:
\begin{equation}
    \frac{d \ket{\psi(t)}}{d t} = -iH\ket{\psi(t)}
\end{equation}
The quantum state at every $t$ can thus be acquired by evolving the quantum state $\ket{\psi(0)}$ by the unitary $U=e^{-iHt}$ as:
\begin{equation}
    \ket{\psi(t)} = e^{-iHt}\ket{\psi(0)}
\label{eq:solution_schrodinger_equation}
\end{equation}
ITE allows the time $t$ to be imaginary. By setting $\tau = it$, with $\tau \in \mathbb{R}$, Eq. \eqref{eq:solution_schrodinger_equation} is written as:
\begin{equation}
    \ket{\psi(\tau)} = e^{-H\tau}\ket{\psi(0)}.
\end{equation}
Clearly, the vector $\ket{\psi(\tau)}$ isn't normalized:
\begin{equation}
    \bra{\psi(\tau)}\ket{\psi(\tau)} = \bra{\psi(0)}e^{-2H\tau}\ket{\psi(0)} \neq 1.
\end{equation}
However, we can enforce normalization:
\begin{equation}
    \ket{\psi(\tau)} = \frac{e^{-H\tau}\ket{\psi(0)}}{\norm{e^{-H\tau}\ket{\psi(0)}}}.
\label{eq:ITE_state}
\end{equation}
We will call the state in Eq. \eqref{eq:ITE_state}, \emph{ITE} state. By taking the derivative with respect to $\tau$, we can show that the dynamics of the ITE state are governed by the Wick-Schrödinger equation:
\begin{equation}
    \frac{d\ket{\psi(\tau)}}{d\tau} = -(H-E(\tau))\ket{\psi(\tau)}
\label{eq:wick_schrodinger_rotation}
\end{equation}
where $E(\tau) = \bra{\psi(\tau)}H\ket{\psi(\tau)}$ is the expectation value of the Hamiltonian for the state $\ket{\psi(\tau)}$. ITE has several important properties that make its application appealing. First of all, consider a state with non-zero overlap with the ground state. Then, applying ITE on that state will converge to the ground state at time $\tau \to \infty $. To see this, consider a state $\ket{\psi(0)}$ that has non-zero amplitude $a_0$ with the ground state
    \begin{equation}
        \ket{\psi(0)} = a_0 \ket{\psi_0} + \sum_{j\neq 0}a_j\ket{\psi_j}
    \end{equation}
and consider the Hamiltonian eigenvalues $0\leq E_0< E_1<\ldots$. Evolving for time $\tau$ results in:
    \begin{gather*}
        \ket{\psi(\tau)} = A(\tau)[a_0e^{-H\tau}\ket{\psi_0} + \sum_{j\neq 0}a_j e^{-H\tau}\ket{\psi_j}] \\
        =A(\tau)e^{-E_0\tau}[a_0 \ket{\psi_0} + \sum_{j\neq 0}a_j e^{-(E_j-E_0)\tau}\ket{\psi_j}]
    \end{gather*}
where $A(\tau) = \norm{e^{-H\tau}\ket{\psi(0)}}^{-1}$. As $E_j-E_0>0$, in the limit of $\tau \rightarrow \infty$ we get $\lim_{\tau \rightarrow \infty} \ket{\psi(\tau)} = \ket{\psi_0}$.

Moreover, starting from the maximally mixed state $\rho(0) = \frac{\mathds{1}}{2^n}$ and performing ITE for time $\tau/2$, will generate a Gibbs (thermal) state:
\begin{equation}
     \frac{e^{-H\tau/2}\rho(0)e^{-H\tau/2}}{\Tr[e^{-H\tau/2}\rho(0)e^{-H\tau/2}]} = \frac{e^{-H\tau}}{\Tr[e^{-H\tau}]} = \frac{e^{-H\tau}}{\sum_je^{-E_j\tau}}
\end{equation}

\subsection{Quantum Algorithms for ITE}

 One of the most popular techniques for applying ITE on a quantum computer is QITE \cite{motta2020determining, gomes2020efficient, nishi2021implementation, huang2023efficient, yeter2022quantum}. Consider a Hamiltonian decomposed into $k$-local Pauli strings $H = \sum_m c_m P_m$. QITE works by first Trotterizing the ITE dynamics as:
\begin{equation}
    e^{-\tau H} = (e^{-c_1P_1\Delta\tau}e^{-c_2P_2\Delta\tau}\ldots)^n + \mathcal{O}(\Delta\tau); \; n:=\frac{\tau}{\Delta\tau}
\end{equation}
and then by searching for a unitary operator $e^{-iA_j\Delta\tau}$ such that:
\begin{equation}
    e^{-iA_j\Delta\tau}\ket{\psi} \approx \frac{e^{-c_jP_j\Delta\tau}\ket{\psi}}{\norm{e^{-c_jP_j\Delta\tau}\ket{\psi}}}
\end{equation}
for each Trotter step. Each $A_j$ is obtained by tomography; first by decomposing $A_j = \sum_{\ell}a_{\ell}P_{\ell}$ into a sum of Pauli strings and then solving a linear system to determine the coefficients $a_{\ell}$. The biggest bottleneck is that the tomographic procedure requires time that scales as $\mathcal{O}(4^D)$ where $D$ is the number of qubits acted on by $A_j$, which can be prohibitively large for strongly correlated Hamiltonians, or for large evolution times $\tau$.

Another approach to execute ITE on a quantum device is to use a parameterized quantum circuit, and employ McLachlan's variational principle to update the parameters \cite{mcardle2019variational, gacon2024variational, kolotouros2025accelerating, gomes2021adaptive}. By doing so, the underlying parameterized unitary generates states that follow the underlying ITE trajectory (i.e., the states that are solutions to the Wick-Schrödinger equation). However, in order to derive the parameter $\boldsymbol{\theta}(\tau)$ updates, one solves the differential equation 
\begin{equation}
    \mathcal{F}_Q(\boldsymbol{\theta}(\tau))\dot{\boldsymbol{\theta}}(\tau) = -2\grad_{\boldsymbol{\theta}}E_{\tau}(\boldsymbol{\theta}(\tau))
\end{equation}
where $\mathcal{F}_Q(\boldsymbol{\theta}(\tau))$ is the quantum Fisher information matrix (QFIM) \cite{liu2020quantum} at point $\boldsymbol{\theta}(\tau)$, and $E_{\tau}(\boldsymbol{\theta}(\tau))$ is the energy of the system at the same configuration. Estimating QFIM can become impractical as the number of circuit parameters increases. Several methods have been developed to reduce total computation resources \cite{gacon2024variational, kolotouros2025accelerating}, but questions about how to choose good initial parameters and how to avoid the vanishing gradients problem \cite{larocca2025barren} remain.

Furthermore, there exist probabilistic approaches \cite{ kosugi2022imaginary, nishi2023optimal, nishi2025encoded, leadbeater2024non, chan2023simulating, zhang2025quantum, silva2023fragmented, huang2024probabilistic} in which the non-unitary operator is encoded in a block of a larger unitary operator. In \cite{kosugi2022imaginary, nishi2023optimal, nishi2025encoded, leadbeater2024non}, the authors showed how to approximate the non-unitary operator with a single ancilla and two controlled real-time evolutions. Moreover, \cite{silva2023fragmented, chan2023simulating,zhang2025quantum} used Quantum Signal Processing (QSP) based approaches to approximate the exponential $e^{-\tau H}$. As we discuss in Appendix \ref{app:comparison_with_other_methods}, these approaches require more resources than our proposed method, and the bound of the success probability depends on both the non-unitarity of the block-encoded operator, and on the state that it acts upon.

\subsection{Double-bracket Imaginary-time evolution}
\label{sec:DBQITE}

Let $\psi(\tau) = \ketbra{\psi(\tau)}$ be the density matrix for the pure ITE state. The Wick-Schrödinger equation in Eq. \eqref{eq:wick_schrodinger_rotation} can be written as:
\begin{equation}
    \partial_{\tau}\ket{\psi(\tau)} = [\psi(\tau), H]\ket{\psi(\tau)}
\label{eq:DB-QITE}
\end{equation}
or in terms of the density operator $\psi(\tau)$ as:
\begin{equation}
    \frac{\partial \psi(\tau)}{\partial \tau} = [[\psi(\tau), H],\psi(\tau)]
\label{eq:ite_dbf_equation}
\end{equation}
which is in the form of Brockett's double-bracket flow \cite{helmke2012optimization}. In \cite{gluza2026double}, Gluza et al. investigated the equivalence between double-bracket flows \cite{brockett1991dynamical, wegner1994flow} and ITE, and developed a \emph{recursive} quantum algorithm with cooling guarantees that they named double-bracket QITE (DB-QITE). Later on, they showed a connection between Grover's algorithm and ITE, as well as how ITE can reproduce known quantum algorithm subroutines \cite{suzuki2025grover}.
  

The key insight in \cite{gluza2026double} is 
to write Eq. \eqref{eq:DB-QITE} (see Lemma 1 in \cite{mcmahon2025equating}) up to an $O(\tau^2)$ error as:
\begin{equation}
    \ket{\psi(\tau)} =e^{\tau[\psi(0),H]}\ket{\psi(0)} + \mathcal{O}(\tau^2),
\end{equation}
where $e^{\tau[\psi(0),H]}$ is a unitary due to the anti-Hermicity of its exponent. 
Hence, for short time duration $s$, we can approximate the above equation to first order as:
\begin{equation}
    \ket{\psi(s)} =e^{s[\psi(0), H]}\ket{\psi(0)}
\label{eq:DBF-approximation}
\end{equation}
with an error $\mathcal{O}(s^2)$. Thus, the above equation results in a proposal for constructing unitaries that approximate ITE. Using a second-order product formula, we can write:
\begin{equation}
\begin{aligned}
    e^{s[\psi(0), H]} = e^{i\sqrt{s}H}e^{i\sqrt{s}\ketbra{\psi(0)}}&e^{-i\sqrt{s}H}e^{-i\sqrt{s}\ketbra{\psi(0)}} \\
    &+ \mathcal{O}(s^{3/2})
\end{aligned}
\label{eq:first_order_approximation}
\end{equation}
The error $\mathcal{O}(s^{3/2})$ can be further reduced if we use higher-order product formulas \cite{childs2013product}.


Next, we can discretize the total evolution time $\tau$ into timesteps $\{s_1, s_2, \ldots, s_K\}$, so that $\sum_ks_k=\tau$, and use Eq. \eqref{eq:first_order_approximation} to perform imaginary-time evolution. Thus, if we let $U_k$ be the unitary that prepares $\ket{\psi_k}$ as $\ket{\psi_k} = U_k\ket{0}$, we arrive at the recursive formula:
\begin{equation}
    U_{k+1} = e^{i\sqrt{s_k}H}U_k e^{i\sqrt{s_k}\ketbra{0}}U_k^{\dagger}e^{-i\sqrt{s_k}H}U_k
\label{eq:recursion_relation}
\end{equation}
where we neglected the global phase $e^{-i\sqrt{s_k}}$. The authors showed that one step of DB-QITE is guaranteed to increase the fidelity with the ground state by an amount proportional to the initial fidelity with the ground-state and the spectral gap of the target Hamiltonian. However, as seen in Eq. \eqref{eq:recursion_relation}, the unitaries that realize ITE are constructed recursively, requiring depths that grow exponentially with the number of ITE steps. Specifically, in order to perform $k$ iterations of DB-QITE, one needs to apply $\mathcal{O}(3^k)$ distinct real-time evolutions. As such, this approach is only useful if we restrict ourselves to a very small number of iterations since the depth rapidly becomes prohibitively large for near-term devices.

Eq. \eqref{eq:recursion_relation} is the starting point for our algorithm. Recall that at each iteration of DB-QITE, the quantum state is evolved as:
\begin{equation}
    \ket{\psi_{k+1}} =e^{i\sqrt{s_k}H}e^{i\sqrt{s_k}\ket{\psi_{k}}\bra{\psi_k}}e^{-i\sqrt{s_k}H}\ket{\psi_k}
\label{eq:dbqite_iteration}
\end{equation}
We can then write the equation above as:
\begin{gather*}
    \ket{\psi_{k+1}} = e^{i\sqrt{s_k}H}(\mathds{1} + (e^{i\sqrt{s_k}}-1)\ket{\psi_k}\bra{\psi_k})e^{-i\sqrt{s_k}H}\ket{\psi_k} \\ 
    = \ket{\psi_k} + (e^{i\sqrt{s_k}} -1)f_k(\sqrt{s_k}) e^{i\sqrt{s_k}H}\ket{\psi_k} \\
    = (\mathds{1} + (e^{i\sqrt{s_k}} - 1)f_k(\sqrt{s_k})e^{i\sqrt{s_k}H})\ket{\psi_k}
\end{gather*}
where $f_k(t)$ is defined as:
\begin{equation}
    f_k(t) = \bra{\psi_k}e^{-iHt}\ket{\psi_k}.
\label{eq:fk_definition}
\end{equation}
The key insight of our algorithm is that instead of constructing the unitary operator in Eq. \eqref{eq:recursion_relation}, one can apply a linear combination of unitaries, in which the coefficients depend on the imaginary-time evolved state at the previous iteration. As we will explain right below in Sec. \ref{sec:pdbqite}, this motivates a probabilistic double-bracket quantum imaginary-time algorithm that requires exponentially shorter depth than DB-QITE \cite{gluza2026double}, at the cost of introducing a probability of success. 

\section{Probabilistic Double-Bracket Quantum Imaginary-time Evolution}
\label{sec:pdbqite}

\begin{figure}
    \centering
\includegraphics[width=1\linewidth]{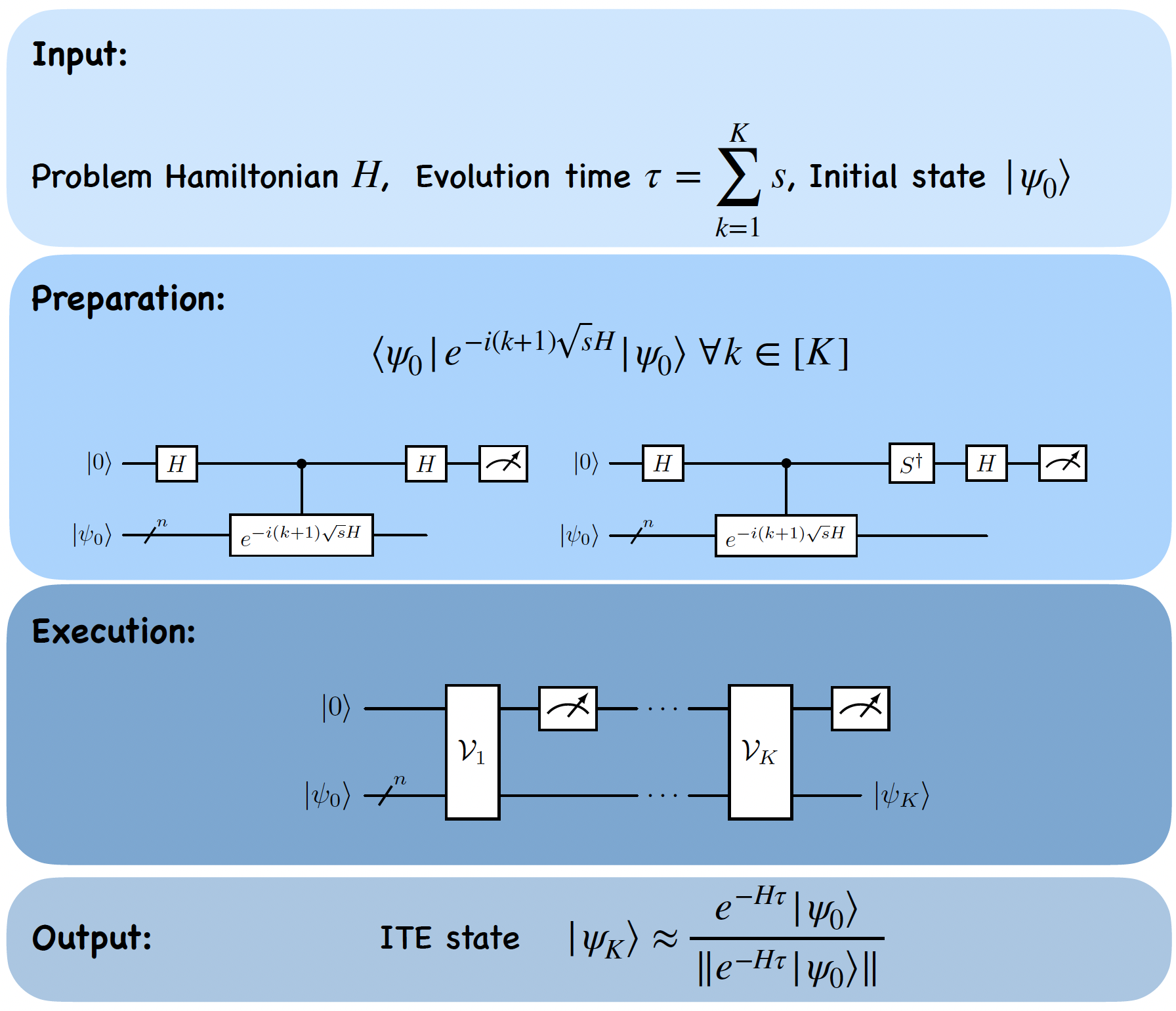}
    \caption{General overview of our method. At first the total evolution time $\tau$ is discretized into $K$ steps. Then, the quantum circuit $\mathcal{V}_k$ corresponding to each step of the algorithm is constructed by measuring the overlap $\bra{\psi_0}e^{-i(k+1)\sqrt{s}H}\ket{\psi_0}$, where $\ket{\psi_0}$ is the initial state. Provided that we measure the ancilla in the state $\ket{0}$ at every step of the algorithm, then at step $K$ we will have an approximation of the ITE state.}
    \label{fig:general_overview_pdbqite}
\end{figure}

In this section, we provide a general overview of our algorithm in \emph{its simplest form}, depicted in Fig. \ref{fig:general_overview_pdbqite}. In later sections, we analyze how the performance of the algorithm can be enhanced and \emph{at what cost} (in terms of total running time, qubits, and gates). 

\subsection{First-order PDBQITE}\label{subsec:1storderPDBITE}

\subsubsection{Input}
 
Our algorithm starts with the problem Hamiltonian $H$, whose ground state we aim to prepare. Similarly to DB-QITE, we discretize the total evolution time $\tau$ into $K$ discrete timesteps $s_k$. The simplest version of the algorithm chooses equal timesteps $s$ and as such $\tau = \sum_k s$, 
which significantly reduces the computational time. We assume an easy-to-prepare initial state $\ket{\psi_0}$ with non-negligible overlap with the ground-state is given at the input of our procedure.

\subsubsection{Preparation}

In order to build the quantum circuit that prepares the ITE state in Eq. \eqref{eq:ITE_state} for time $\tau$, we need to estimate the overlaps:
\begin{equation}   \bra{\psi_0}e^{-i(k+1)\sqrt{s}H}\ket{\psi_0}
\label{eq:overlaps_initial_state}
\end{equation}
for all $k\in [K]$. These overlaps can be estimated using two Hadamard tests per iteration for a total of $2K$ Hadamard tests. As we show in Appendix \ref{sec:estimating_fk}, these overlaps can be used to calculate the quantities $f_k(\sqrt{s})$ (of intermediate states $\ket{\psi_k}$) in Eq. \eqref{eq:fk_definition} by performing measurements only on the initial state $\ket{\psi_0}$. However, as we discuss in Appendix \ref{sec:estimating_fk}, this can only be done efficiently iff $s_k=s$ for all $k\in [K]$. In other cases where the timesteps are different, one would have to perform intermediate measurements on the states $\ket{\psi_k}$ to infer the quantities $f_k(\sqrt{s_k})$.

\subsubsection{Execution}

Once we have estimated the overlaps in Eq. \eqref{eq:overlaps_initial_state}, we proceed to the main part of the algorithm. Recall that the action of one DB-QITE can be expressed as:
\begin{equation}
\begin{gathered}
    \ket{\psi_{k+1}} = \mathcal{Q}_k\ket{\psi_k}
\end{gathered}
\end{equation}
where:
\begin{equation}
    \mathcal{Q}_k = (\mathds{1} + (e^{i\sqrt{s}}-1) f_k(\sqrt{s})e^{i\sqrt{s}H})
\end{equation}
is a non-unitary operator that can be expressed as a linear combination of unitaries. In Sec. \ref{sec:main_results}, we thoroughly explain how to implement this non-unitary operator using a single ancilla qubit and a probabilistic approach. Specifically, $\mathcal{V}_k$ is the LCU unitary that generates the state:
\begin{equation}
    \mathcal{V}_k\ket{0}\ket{\psi_k} = \frac{1}{\alpha_k}\ket{0}\mathcal{Q}_k\ket{\psi_k} + \sqrt{1-\frac{1}{\alpha_k^2}}\ket{\Phi^{\perp}}
\end{equation}
where $\ket{\Phi^{\perp}}$ corresponds to the failure state. This implies that the system has performed imaginary-time evolution for a timestep $s$, provided that the ancilla qubit is measured in the $\ket{0}$ state. The specific form of this non-unitary operator allows for a lower bound on the success probability given in Lemma \ref{lemma:prob_of_success_first_order} as:
\begin{equation}
    P_{\text{succ}}\geq \frac{1}{(1+\sqrt{s})^2}
\label{eq:probability_of_success_lower_bound}
\end{equation}
Eq. \eqref{eq:probability_of_success_lower_bound} indicates that the probability of success is close to unity for small timesteps $s$. As a result, we can perform a moderate number of small timesteps $s$, while maintaining a relatively high probability of success. 

In order to apply $K$ consecutive iterations of PDBQITE, we need to apply a series of LCU unitaries $\{\mathcal{V}_1,\mathcal{V}_2, \ldots, \mathcal{V}_K\}$ sequentially, provided that we measure the ancilla qubit in the state $\ket{0}$ after every iteration.

\subsubsection{Output}

After executing the circuit, and measuring the ancilla in the state $\ket{0}$ at every step, the system will be in the state:
\begin{equation}
    \ket{\psi_K}\approx \frac{e^{-\tau H}\ket{\psi_0}}{\norm{e^{-\tau H}\ket{\psi_0}}}
\end{equation}
We outline the pseudoalgorithm for PDBQITE in Algorithm \ref{alg:PDB-QITE} and note that this is our algorithm in its simplest form. In later sections, we describe how to enhance the performance of PDBQITE at the cost of more ancillas, measurements, and gates. Specifically, in Sec. \ref{sec:second-order-pdbqite} we explain how to use an additional ancilla qubit that allows the discretization of the total time $\tau$ into larger timesteps.

\begin{algorithm}[h!]
\caption{PDBQITE}
\label{alg:PDB-QITE}
\SetKwInOut{Input}{Input}
\Input{Problem Hamiltonian $H$\;
Initial state $\ket{\psi_0}$\;
Total evolution time $\tau$\;
Number of iterations $K$\;
Overlaps $\bra{\psi_0}e^{-i(k+1)\sqrt{s}H}\ket{\psi_0}$ for $k\in[K]$\;
Estimate $f_1, f_2, \ldots, f_K$\;}
\For{$k=1,2,\ldots,K$}{

Construct and apply LCU circuit $\mathcal{V}_k$\;
Measure the ancilla qubit in the $\{\ket{0},\ket{1}\}$ basis\;
\If{Ancilla is in $\ket{1}$}{\KwSty{break}}
Postselect $\ket{\psi_{k}}$ with probability $P_{\text{succ}}$ by measuring the ancilla qubit in the $\ket{0}$ state.\;
}
\Return $\ket{\psi_K}  \approx \frac{e^{-\tau H}\ket{\psi_0}}{\norm{e^{-\tau H}\ket{\psi_0}}}$
\end{algorithm} 

\subsection{Further Improvements}

Instead of using the second-order product formula of Eq. \eqref{eq:first_order_approximation}, one could use a third-order product formula to approximate the exponential $e^{s[\psi(0), H]}$ with greater accuracy and permitting larger time steps, as summarized in Sec. \ref{sec:second-order-pdbqite} and derived in Appendix \ref{app:higher_order_formulas}.
Its limited additional cost of a second ancillary qubit and one more controlled-real-time evolution can nonetheless lead to improvements of overall circuit depth and total success probability due to the savings in the number of iterations required to reach the target state, as demonstrated in the discussion of our numerical results in section \ref{sec:experiments}.

The operator $\mathcal{Q}_k$ is close to a unitary operator, up to an error $\mathcal{O}(\sqrt{s_k})$. When this operator acts on the state $\ket{\psi_k}$, it succeeds with a probability that is large and close to 1. This motivates a robust (partial) oblivious amplitude amplification technique, inspired by the work in \cite{berry2015hamiltonian}, that can boost the per-step probability of success to even larger values, as we present in Appendix \ref{app:amplitude_amplification}.

\section{Numerical Experiments}
\label{sec:experiments}

In this section, we investigate how well PDBQITE works in practice. The performance and convergence of ITE approaches is highly dependent on the initial state. Specifically, states that have a non-vanishing (polynomial) overlap with the ground state (or with a low-energy state) tend to converge much faster than starting from the uniform superposition state. Our analysis will allow us to demonstrate how PDBQITE can outperform PITE, one of the most popular probabilistic quantum ITE methods so far. Moreover, we will compare the performance of first- and second-order PDBQITE and show how PDBQITE can enhance the performance of QAOA in solving Max-Cut instances. Appendix \ref{sec:additional_experiments} contains additional experiments on the Maximum Independent Set problem, reaching similar conclusions.

\subsection{Ground-state preparation}
\begin{figure*}[t]
    \centering    \includegraphics[width=0.85\textwidth]{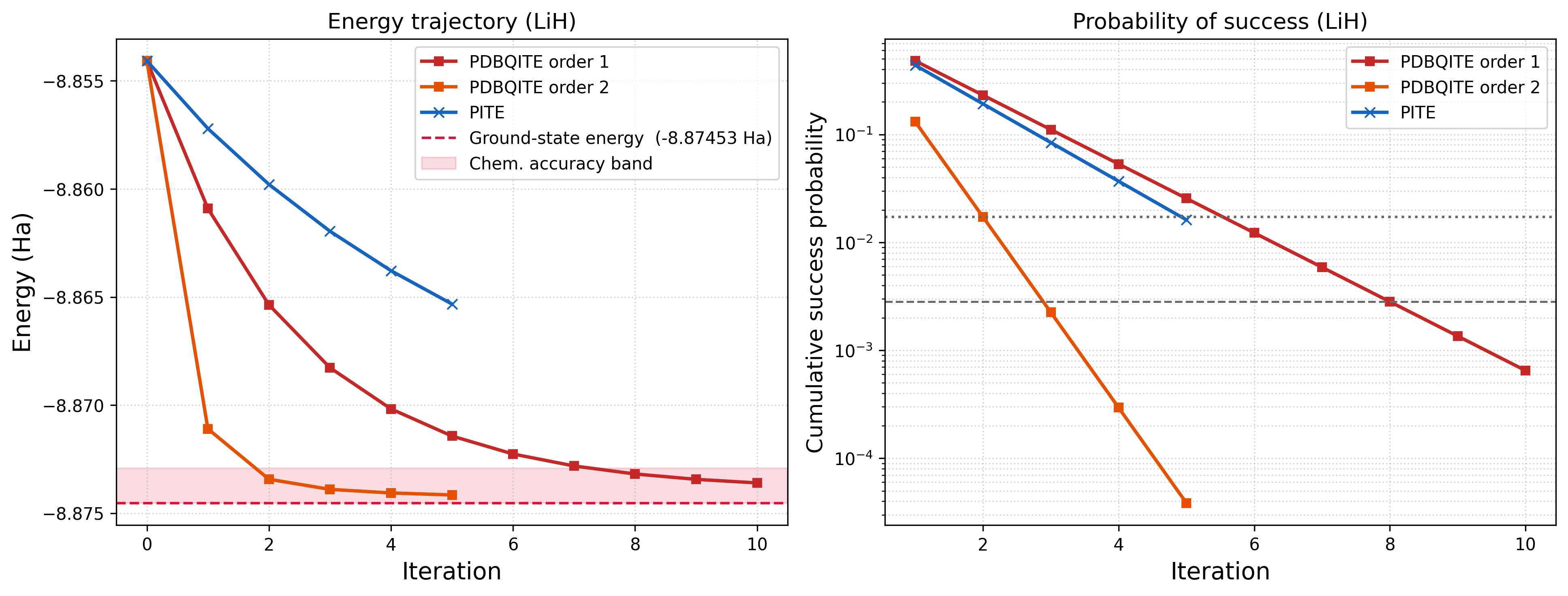} \\
\caption{Energy trajectory (left figure) and cumulative probability of success (right figure) for the $LiH$ molecule. For first-order PDBQITE we choose a timestep $s_k=0.2$ and perform 10 iterations, while for second-order PDBQITE and PITE we choose a timestep $s_k=0.7$ and perform only 5 iterations. Horizontal lines mark the cumulative success probability at which second-order (dotted) and first-order (dashed) PDBQITE reach chemical accuracy.}
    \label{fig:molecular_ground_state_prep}
\end{figure*}

Molecular systems are a standard benchmark for quantum ground-state preparation algorithms, due to their hardness as the size of molecules increases \cite{mcardle2019variational, bauer2020quantum, kirby2026observation, nishimura2026symmetry}. A molecule with $N$ electrons in positions $\boldsymbol{r}_i$ and $M$ nuclei in positions $\boldsymbol{R}_I$ with charge $Z_I$ has a Hamiltonian expressed in atomic units as:
\begin{equation}
\begin{gathered}
    H = -\sum_{i=1}^N\frac{\grad_i^2}{2} - \sum_{I=1}^M\frac{\grad_I^2}{2M_I}-\sum_{i=1}^N\sum_{I=1}^M\frac{Z_I}{|\boldsymbol{r}_i-\boldsymbol{R}_I|}+\\
    \sum_{i=1}^N\sum_{j=i+1}^N\frac{1}{|\boldsymbol{r_i}-\boldsymbol{r}_j|} +\sum_{I=1}^M\sum_{J=I+1}^M\frac{Z_IZ_J}{|\boldsymbol{R}_I-\boldsymbol{R}_J|}
\end{gathered}
\end{equation}
Following the standard methodology, we use the Born-Oppenheimer approximation that assumes that the masses of the atomic nuclei are significantly larger than the mass of the electron. As such, we can express the total wavefunction as the product of an electronic and a nuclear wavefunction. The electronic Hamiltonian can then be expressed in second quantization as:
\begin{equation}
    H_{\text{elec}} = \sum_{pq}h_{pq}\hat{a}_p^{\dagger}\hat{a}_q + \frac{1}{2}\sum_{pqrs}h_{pqrs}\hat{a}_p^{\dagger}\hat{a}_q^{\dagger}\hat{a}_r\hat{a}_s
\end{equation}
where the 1-body integrals $h_{pq}$ and 2-body integrals $h_{pqrs}$ can be calculated efficiently using classical codes.

In our experiments, we choose to tackle two different molecules $\mathrm{LiH}$ and $\mathrm{BeH_2}$,
using both first- and second-order PDBQITE and with the Hartree-Fock state as the initial state. The results for $\mathrm{LiH}$ are visualized in Fig. \ref{fig:molecular_ground_state_prep}, while the results for $\mathrm{BeH_2}$, found in Appendix \ref{sec:additional_experiments}, follow a similar pattern.  

A non-trivial question when using PDBQITE is to decide whether to use first or second order to solve a given problem. As shown in Figure \ref{fig:molecular_ground_state_prep} (left),
second-order PDBQITE converges faster to the ground-state in terms of the number of iterations. The key question is whether it remains favorable in terms of success probability and gate count, or circuit depth. Figure \ref{fig:molecular_ground_state_prep} (right) illustrates how the probability of success for second-order PDBQITE decreases faster than that of first order, due to a smaller probability of success per iteration. Crucially, because second order reaches chemical accuracy at its second step, while first order requires more iterations, the success probability to reach chemical accuracy of second order remains almost one order of magnitude higher than that of the first order.

In terms of gate and circuit depth, the discussion is more subtle. On one hand, every iteration of second-order PDBQITE requires twice the number of controlled real-time evolution calls than that of first-order PDBQITE, which is compensated by reaching chemical accuracy almost at step 2 instead of 8. On the other hand, one needs to take into account also the cost of Trotterization and the different time steps, $s=0.2$ for first-order and $s=0.7$ for second-order, which may lead to higher cost in circuit depth per iteration. Thus, the total circuit depth to reach chemical accuracy is comparable for both orders, while second-order retains its success-probability advantage.


Furthermore, our data show that both versions outperform PITE \cite{kosugi2022imaginary}. While PITE has a similar probability of success to first-order PDBQITE, it performs worse in energy per iteration, while at the same time requiring deeper circuits as it needs two controlled real-time evolution calls per iteration, and each iteration will probably have a higher Trotterization cost due to its larger timestep. PITE and second-order PDBQITE have equivalent circuit depth per iteration, but second-order PDBQITE significantly outperforms in speed of convergence to the ground state, while reaching chemical accuracy with a higher success probability.

\subsection{Classical optimization problems}

\begin{figure*}[t]
    \centering
    \includegraphics[width=\textwidth]{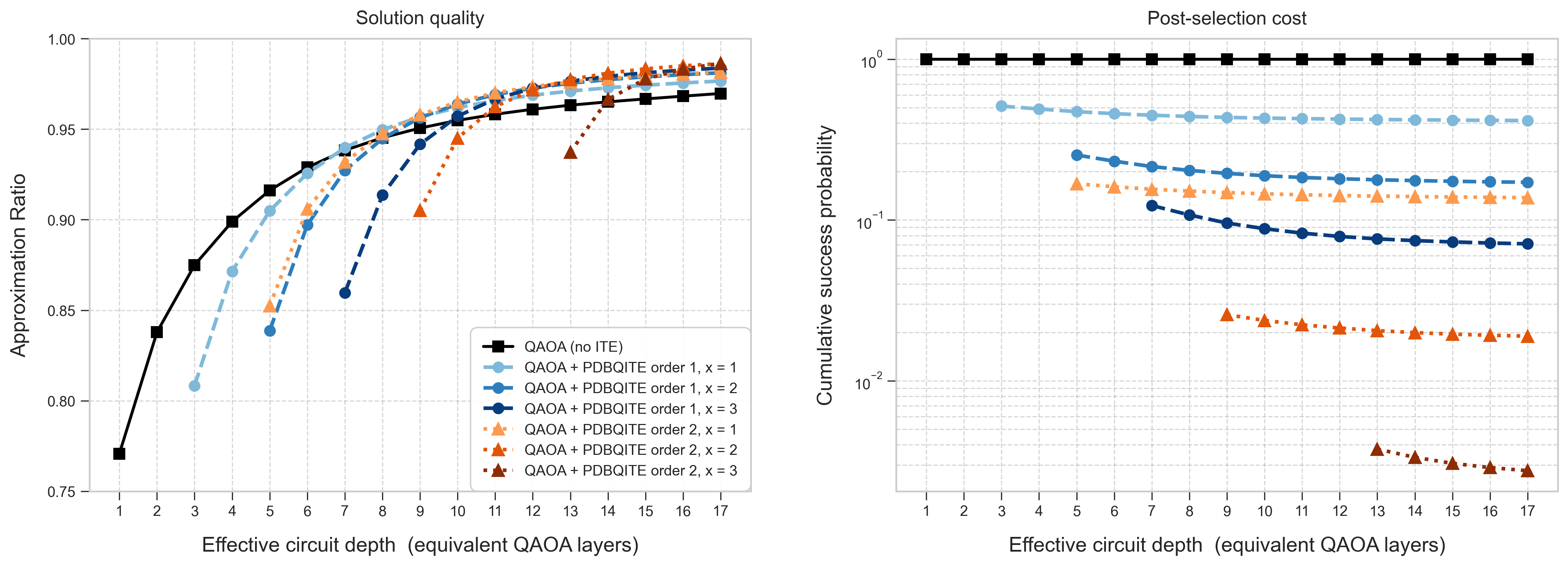} \\
\caption{Approximation ratio (left) and cumulative success probability (right) as a function of ``equivalent QAOA layers'', for randomly weighted 4-regular MaxCut
instances of size $n=18$. We compare QAOA against QAOA augmented with up to three PDBQITE steps of first ($s=0.35$) and second order ($s=0.7$).}
\label{fig:qaoa_vs_qaoaite}
\end{figure*}

In this section, we benchmark PDBQITE on the MaxCut problem on different types of graphs. As a figure of merit, we chose the approximation ratio defined as 
$r = \frac{C_{\text{approx}}}{C_{\text{optimal}}}$, where $C_{\text{approx}}$ is the cost value returned by our algorithm and $C_{\text{optimal}}$ is the optimal cost value. 

Applying PDBQITE directly to the QAOA initial state $\ket{+}^{\otimes n}$, or to any other quantum state with a weak overlap with the ground-state, requires many iterations to produce meaningful results. We therefore investigate a hybrid approach where we first generate a quantum state by applying $p$ layers of QAOA, and then perform $x$ layers of PDBQITE.
To conduct a fair resource estimation comparison, we use the fact that one ITE step for first- and second-order PDBQITE is equivalent, in terms of number of single- and two-qubit gates, to two and four QAOA layers respectively (see Appendix \ref{app:controlled_real_time_evolution} for a derivation). 

For our benchmarks, we use randomly weighted 3-regular and 4-regular graphs. The weights of the graphs were sampled from the uniform distribution $\mathcal{U}(0,1)$ and rounded to 3 decimal points. For the initial state, we choose the quantum state generated by a QAOA \cite{farhi2014quantum} circuit, with parameters $(\boldsymbol{\beta}, \boldsymbol{\gamma})$ chosen according to the SKatan method \cite{vcepaite2025quantum, shaydulin2023parameter}:
$\ket{\psi_0} = U_{\text{QAOA}}^p(\boldsymbol{\beta}, \boldsymbol{\gamma})\ket{+}^{\otimes n}$.
We tested 20 graph instances of size $n=18$ and for up to 17 QAOA layers. More experiments for different graph sizes and different types of graphs can be found in Appendix \ref{sec:additional_experiments}.

Fig. \ref{fig:qaoa_vs_qaoaite} shows the approximation ratio as a function of equivalent QAOA layers. For $x$ iterations of $y$-th order PDBQITE, after $p$ initial layers of QAOA, the circuit depth $D$ measured in equivalent layers of standard QAOA reads: $D=p+2yx$. We observe that until 6 to 7 equivalent layers, there is no benefit in using any PDBQITE iteration over vanilla QAOA. Interestingly, beyond that threshold, slowly adding iterations of first-order PDBQITE improves the ratio, at the cost of a success probability greater than $\sim 0.1$.

 Having access to additional circuit depth opens up the possibility of using second-order PDBQITE that can slightly increase performance, but at a cost in success rate that may not be worth it. For first-order PDBQITE, we used a timestep $s=0.35$, while for second-order PDBQITE we used a larger timestep, set to $s=0.7$. These naively chosen timesteps perform rather well in practice, a result that could only improve if we optimize the timestep choices.

\section{First-order PDBQITE}
\label{sec:main_results}




 We start by showing how we can build these DB-QITE circuits using quantum circuits with depth that grows \emph{linearly} with the number of ITE iterations, providing an exponential advantage in depth compared to the original DB-QITE circuit synthesis. In this section, we will analyze the general case in which the time steps are not necessarily equal. Recall that at each iteration of DB-QITE the state is evolved as:
\begin{equation}
    \ket{\psi_{k+1}} =e^{i\sqrt{s_k}H}e^{i\sqrt{s_k}\ket{\psi_{k}}\bra{\psi_k}}e^{-i\sqrt{s_k}H}\ket{\psi_k}
\label{eq:dbqite_single_iteration_state}
\end{equation}
up to a global phase $e^{-i\sqrt{s_k}}$ that is neglected. As we also discussed in Sec. \ref{sec:DBQITE}, instead of following the approach of Gluza et al. \cite{gluza2026double}, we can rewrite the equation above as:
\begin{gather*}
    \ket{\psi_{k+1}} = e^{i\sqrt{s_k}H}(\mathds{1} + (e^{i\sqrt{s_k}}-1)\ket{\psi_k}\bra{\psi_k})e^{-i\sqrt{s_k}H}\ket{\psi_k} \\ 
    = \ket{\psi_k} + (e^{i\sqrt{s_k}} -1)f_k(\sqrt{s_k}) e^{i\sqrt{s_k}H}\ket{\psi_k} \\
    = (\mathds{1} + (e^{i\sqrt{s_k}} - 1)f_k(\sqrt{s_k})e^{i\sqrt{s_k}H})\ket{\psi_k}
\end{gather*}
where $f_k(t)$ is defined as:
\begin{equation}
    f_k(t) = \bra{\psi_k}e^{-iHt}\ket{\psi_k}.
\end{equation}

As a result, we have written the action in Eq. \eqref{eq:dbqite_single_iteration_state} as a linear combination of two unitaries, one being the identity and the other the real-time (backward) evolution for time $\sqrt{s_k}$. Let the operator $\mathcal{Q}_k$, defined as the linear combination of the two unitaries:
\begin{equation}
    \mathcal{Q}_k = \mathds{1} + 
    (e^{i\sqrt{s_k}} - 1)f_k(\sqrt{s_k})e^{i\sqrt{s_k}H}
\label{eq:linear_combination_unitary}
\end{equation}
The operator $\mathcal{Q}_k$ is not a unitary operator, but produces a normalized state when it acts on the state $\ket{\psi_k}$. 

\begin{lemma}
    The action of the operator $\mathcal{Q}_k$ 
    \begin{equation*}
    \mathcal{Q}_k = \left(\mathds{1} + 
    (e^{i\sqrt{s_k}} - 1)f_k(\sqrt{s_k})e^{i\sqrt{s_k}H}\right)
    \end{equation*}
    on the state $\ket{\psi_k}$ results in a quantum state $\ket{\psi_{k+1}}$, where $\norm{\ket{\psi_{k+1}}}^2 = \bra{\psi_k}\mathcal{Q}_k^{\dagger}\mathcal{Q}_k\ket{\psi_k} = 1$.
\label{lemma:isometry_on_psi_k}
\end{lemma}

\begin{proof}
    Consider the action of the operator $\mathcal{Q}_k$ on the state $\ket{\psi_k}$, which results in the state $\ket{\psi_{k+1}}\equiv\mathcal{Q}_k\ket{\psi_k}$. Let also $\Delta = (e^{i\sqrt{s_k}}-1)$. If we calculate the squared norm of $\ket{\psi_{k+1}}$ we get:
    \begin{gather*}
        \bra{\psi_{k+1}}\ket{\psi_{k+1}}  \\
        = 1 + \Delta f_k(\sqrt{s_k})\bra{\psi_k}e^{i\sqrt{s_k}H}\ket{\psi_k} \\+ \Delta^*f_k(\sqrt{s_k})^*\bra{\psi_k}e^{-i\sqrt{s_k}H}\ket{\psi_k} + |\Delta|^2|f_k(\sqrt{s_k})|^2\\
        =1 + |f_k(\sqrt{s_k})|^2(\Delta + \Delta^* + |\Delta|^2)
    \end{gather*}
If we evaluate the term in the parenthesis we get:
\begin{gather*}
    \Delta + \Delta^* + |\Delta|^2 = \\
    (e^{i\sqrt{s_k}}-1) + (e^{-i\sqrt{s_k}}-1) + |e^{i\sqrt{s_k}}-1|^2=0
\end{gather*}
Thus, we have proved that $\bra{\psi_{k+1}}\ket{\psi_{k+1}}=1$.
\end{proof}

In order to apply the non-unitary operator $\mathcal{Q}_k$ on the state $\ket{\psi_k}$, we will use the Linear Combination of Unitaries (LCU) approach \cite{childs2012hamiltonian}. 

\subsection{Circuit construction}

Fig. \ref{figure:single_iteration_pdbqite} illustrates the circuit to implement one iteration of first-order \emph{probabilistic} double-bracket QITE (PDBQITE). Being an LCU circuit, it consists of three steps, $\mathsf{PREP}$, $\mathsf{SELECT}$ and a final projection on the $\mathsf{PREP}$ state.

$\mathsf{PREP}$ starts by setting the ancilla qubit in the state $\ket{0}$ and applying to it the $U_{\mathsf{PREP}}$ circuit that encodes the coefficients of the linear combination:
\begin{equation}
    U_{\mathsf{PREP}}\ket{0} = \frac{1}{\sqrt{\alpha_k}}\Big(\ket{0} + \sqrt{|a_k|}\ket{1}\Big)
\end{equation}
where $\alpha_k = 1+|a_k|$ is the sum of the magnitude of both coefficients in the linear combination of the $k$-th iteration, i.e., $a_k= (e^{i\sqrt{s_k}} - 1)f_k(\sqrt{s_k})$. The unitary $U_{\mathsf{PREP}}$ can easily be constructed by applying an $R_y(\theta)$ rotation with $\theta = 2 \arctan\sqrt{|a_k|}$. Because $a_k$ is a complex number, we also need to encode its angle $\phi$, defined as $a_k= |a_k|e^{i\phi}$, by applying an $R_z(\phi)$ gate on the ancilla qubit. This creates the state (up to a global phase):
\begin{equation}
    R_z(\phi)U_{\mathsf{PREP}}\ket{0} = \frac{1}{\sqrt{1+|a_k|}}\Big(\ket{0} +e^{i\phi}\sqrt{|a_k|}\ket{1}\Big)
\label{eq:action_U_prep_and_rz}
\end{equation}

The second stage applies the $\mathsf{SELECT}$ operator $U_{\mathsf{SEL}}$ on the total system (including the ancilla). The $\mathsf{SELECT}$ operator is defined so that it applies the corresponding unitary in state $\ket{\psi_k}$ depending on the value of the ancilla register. As such $U_{\mathsf{SEL}}$ is defined as:
\begin{equation}
    U_{\mathsf{SEL}} = \ketbra{0}\otimes \mathds{1} + \ketbra{1}\otimes e^{i\sqrt{s_k}H}
\end{equation}
which is a controlled-$(e^{i\sqrt{s_k}H})$ (i.e., a controlled backward real-time evolution). Applying $U_{\mathsf{SEL}}$ on the state in Eq. \eqref{eq:action_U_prep_and_rz} results in:
\begin{equation}
    \begin{aligned}
U_{\mathsf{SEL}}R_z(\phi)U_{\mathsf{PREP}}\ket{0}\ket{\psi_k} \\ =\frac{1}{\sqrt{1+|a_k|}}\Big(U_{\mathsf{SEL}}\ket{0}\ket{\psi_k} +e^{i\phi}\sqrt{|a_k|}U_{\mathsf{SEL}}\ket{1}\ket{\psi_k}\Big) \\
= \frac{1}{\sqrt{1+|a_k|}}\Big(\ket{0}\ket{\psi_k} +e^{i\phi}\sqrt{|a_k|}\ket{1}e^{i\sqrt{s_k}H}\ket{\psi_k}\Big)
\end{aligned}
\end{equation}

The final step, is the projection into the $\mathsf{PREP}$ state, i.e.,
\begin{equation}
  \bra{0}U_{\mathsf{PREP}}^{\dagger} = \frac{1}{\sqrt{\alpha_k}}\Big(\bra{0} + \sqrt{|a_{k}|}\bra{1}\Big),
\end{equation}
via acting with $U_{\mathsf{PREP}}^{\dagger} = R_y(-\theta)$ and measuring in the computational basis,
leading to the following subnormalised output:
\begin{equation}
\begin{aligned}
(\bra{0}\otimes\mathds{1})U_{\mathsf{PREP}}^{\dagger}U_{\mathsf{SEL}}R_z(\phi)U_{\mathsf{PREP}}\ket{0}\ket{\psi_k} \\= \frac{1}{\alpha_k}\Big(\mathds{1} + a_ke^{i\sqrt{s_k}H}\Big)\ket{\psi_k}.
\end{aligned}
\end{equation}

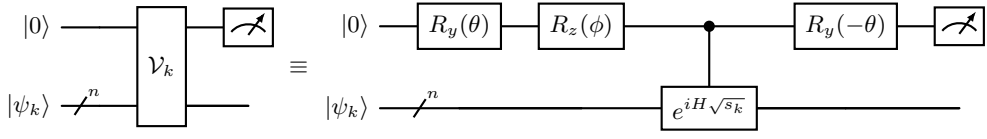
\begin{figure*}
\begin{center}
\begin{quantikz}
    \lstick{$\ket{0}$} & &\gate[2]{\mathcal{V}_k} &\meter{} \\
    \lstick{$\ket{\psi_k}$} &\qwbundle{n}& &\qw
\end{quantikz}
$\equiv$
\begin{quantikz}
\lstick{$\ket{0}$} & \gate{R_y(\theta)} & \gate{R_z(\phi)} & \ctrl{1} & \gate{R_y(-\theta)} & \meter{} \\
\lstick{$\ket{\psi_k}$} & \qwbundle{n} & \qw & \gate{e^{iH\sqrt{s_k}}} & \qw & \qw
\end{quantikz}
\end{center}
\caption{Linear combination of unitaries circuit to perform one iteration of probabilistic DB-QITE. If the ancilla qubit is measured in the $\ket{0}$ state, the remaining system has performed imaginary-time evolution.}
\label{figure:single_iteration_pdbqite}
\end{figure*}

In order to construct the operator $\mathcal{Q}_k$, one needs to estimate the complex coefficient $f_{k}(\sqrt{s_k})=\bra{\psi_k}e^{-iH\sqrt{s_k}}\ket{\psi_k}$, which requires a Hadamard test subroutine on the state at the previous iteration (see Appendix \ref{sec:estimating_fk} for more details). 
Interestingly, as shown in Appendix \ref{sec:estimating_fk}, if the timesteps are chosen to be \emph{constant} ($s_k = s$ for all $k$) throughout the evolution, then all $f_k$ can be estimated in advance from the time-series
$\bra{\psi_0}e^{-iH\ell\sqrt{s_k}}\ket{\psi_0}$, with $\ell \in [K]$, that can be computed by a similar Hadamard-test subroutine now acting only on the input state $\ket{\psi_0}$.

Hence, constant steps simplify the algorithm, as presented in Sec. \ref{subsec:1storderPDBITE}. Non-uniform steps, on the other hand, allow to optimize further the speed of convergence, but require measuring intermediate states. Finally, in some situations where time steps are small, it could be beneficial to estimate $f_k$ by measuring moments of the Hamiltonian instead of using a Hadamard-test subroutine. We point the reader to Appendix \ref{sec:estimating_fk} for more information.

\subsection{Success probability}

A crucial difference of PDBQITE with respect to the original DB-QITE proposal is the probabilistic nature of our algorithm. While this does not pose a problem when we restrict ourselves to a moderate number of PDBQITE iterations, it can be limiting if we want to perform a large number of them, a recurrent issue for other ITE probabilistic approaches such as \cite{suzuki2025double, kosugi2022imaginary, xie2024probabilistic, yi2025probabilistic}. An important advantage of our approach is that the success rate of every step is relatively high (increasing as the time step becomes smaller), is completely independent of the system size, and can be lower-bounded, as we show in
Lemmas \ref{lemma:prob_of_success_exact_first_order} and \ref{lemma:prob_of_success_first_order}.

\begin{lemma}
    Let the system be in the state $\ket{\psi_k}$ and consider the action of the LCU unitary that implements the non-unitary operator $\mathcal{Q}_k = \mathds{1} + 
    (e^{i\sqrt{s_k}} - 1)f_k(\sqrt{s_k})e^{i\sqrt{s_k}H}$. The probability of measuring the ancilla in the state $\ket{0}$ is:
    \begin{equation}
    P_{\text{succ}} = \Bigg(\frac{1}{1+2\left|\sin(\frac{\sqrt{s_k}}{2})\right|\left|\bra{\psi_{k}}e^{-i\sqrt{s_k}H}\ket{\psi_k}\right|}\Bigg)^2
    \label{eq:probability_of_success_exact_form}
    \end{equation}
\label{lemma:prob_of_success_exact_first_order}
\end{lemma}
\begin{proof}
    The probability of applying the LCU circuit and measuring the ancilla qubit to be in the state $\ket{0}$ is:
    \begin{equation}
    \begin{aligned}
        P_{\text{succ}} = \frac{\bra{\psi_k}\mathcal{Q}_k^\dagger \mathcal{Q}_k\ket{\psi_k}}{\alpha_k^2} = \frac{1}{\alpha_k^2}
    \end{aligned}
    \end{equation}
where we used the fact that the non-unitary operator $\mathcal{Q}_k$ generates a normalized state when it acts on $\ket{\psi_k}$ (see Lemma \ref{lemma:isometry_on_psi_k}). Next, we have that:
\begin{equation}
    \alpha_k^2 = (1 + |\Delta||f_k(\sqrt{s_k})|)^2
\end{equation}
where $\Delta = (e^{i\sqrt{s_k}}-1)$. We can expand $|\Delta|$:
\begin{gather*}
    |\Delta|^2 = (e^{i\sqrt{s_k}}-1)(e^{-i\sqrt{s_k}}-1)  = 4\sin^2\Big(\frac{\sqrt{s_k}}{2}\Big)
\end{gather*}
Hence, $|\Delta| = 2\Big|\sin(\frac{\sqrt{s_k}}{2})\Big| $, and the probability of success can be written as:
\begin{equation*}
    P_{\text{succ}} = \Bigg(\frac{1}{1+2\left|\sin(\frac{\sqrt{s_k}}{2})\right|\left|\bra{\psi_{k}}e^{-i\sqrt{s_k}H}\ket{\psi_k}\right|}\Bigg)^2
\end{equation*}
\end{proof}

The probability of success in Eq. \eqref{eq:probability_of_success_exact_form} can be lower bounded according to Lemma \ref{lemma:prob_of_success_first_order}.

\begin{lemma}
\label{lemma:prob_of_success_first_order}
    The probability of success is lower bounded as:
    \begin{equation}
        P_\text{succ} \geq \frac{1}{(1+\sqrt{s_k})^2}
    \label{eq:probability_of_success}
    \end{equation}
\end{lemma}
\begin{proof}
    Let $g(t) = 1+2\Big|\sin(\frac{\sqrt{t}}{2})\Big|\left|\bra{\psi_{k}}e^{-i\sqrt{t}H}\ket{\psi_k}\right|$. If we use the fact that $|\sin(\frac{\sqrt{t}}{2})\big| \leq \frac{\sqrt{t}}{2}$ for $t\geq 0$ and also that $\left|\bra{\psi_{k}}e^{-i\sqrt{t}H}\ket{\psi_k}\right|\leq 1$ for any unitary operator $e^{-iHt}$, then we have the upper bound $g(t)\leq 1+\sqrt{t}$. Thus:
    \begin{equation}
        P_{\text{succ}} = \frac{1}{g(s_k)^2} \geq \frac{1}{(1+\sqrt{s_k})^2}
    \end{equation}
\end{proof}
Note that tighter upper bounds on the amplitude $\left|\bra{\psi_{k}}e^{-i\sqrt{s_k}H}\ket{\psi_k}\right|$ can lead to tighter lower bounds on the success probability.

\section{Second-order PDBQITE}
\label{sec:second-order-pdbqite}

One can implement larger imaginary time-steps, and reduce the number of required iterations by using higher-order product formulas to approximate the exponential $e^{s[\psi(0), H]}$ with greater accuracy.

\begin{figure*}[t]
\begin{center}
\begin{quantikz}
\lstick{$\ket{0}$}& \gate[2]{U_{\mathsf{PREP}}} &\gate[2]{U_{\text{phase}}} & \qw       & \ctrl{2} & \qw &\gate[2]{U_{\mathsf{PREP}}^{\dagger}} &\meter{} \\
\lstick{$\ket{0}$}      &                           & & \ctrl{1}                & \qw      & \qw & &\meter{} \\
\lstick{$\ket{\psi_k}$} & \qwbundle{n}              && \gate{e^{-i(1-\varphi)\sqrt{s_k}H}} & \gate{e^{i\varphi\sqrt{s_k}H}} & \qw &&
\end{quantikz}
\end{center}
\caption{LCU circuit to perform one iteration of second-order PDBQITE. Measuring the ancilla qubits in the $\ket{0}^{\otimes 2}$ state will result in one step of imaginary-time evolution (for time $s_k$) for the system on the last $n$-qubit register.}
\label{figure:second_order_LCU}
\end{figure*}
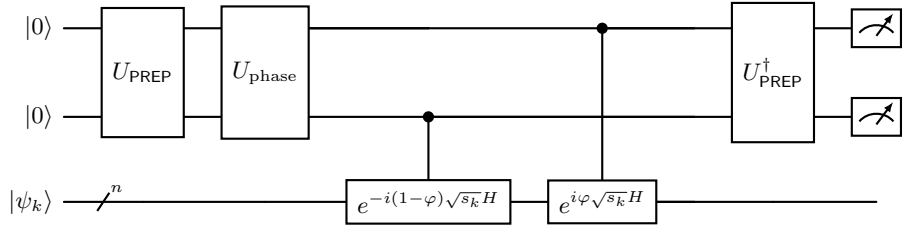

\subsubsection{Second-order DB-QITE}

Using a third-order product formula \cite{chen2022efficient,childs2012hamiltonian} to approximate $e^{s[\psi(0), H]}$, leads to the following recursive relation of unitaries:
\begin{equation}
\begin{gathered}
        V_{k+1} = e^{i\varphi\sqrt{s_k}H}V_ke^{i\varphi \sqrt{s_k}\ketbra{0}}V_k^{\dagger}e^{-i\sqrt{s_k}H}V_k \\e^{-i(\varphi+1)\sqrt{s_k}\ketbra{0}}V_k^{\dagger}e^{i(1-\varphi)\sqrt{s_k}H}V_k
\end{gathered}
\end{equation}
where $\varphi = (\sqrt{5}-1)/2$.

In the traditional DB-QITE algorithmic framework, constructing these unitaries results in a quantum algorithm that requires $\mathcal{O}(5^k)$ real-time evolution operators at each iteration $k$, versus $\mathcal{O}(3^k)$ for the first-order DB-QITE recursion of Eq. (\ref{eq:recursion_relation}), making the shift from first-order to second-order exponentially costly. 

\subsubsection{Second-order PDBQITE}

As derived in Appendix \ref{app:higher_order_formulas}, PDBQITE can exploit the
third-order product formula for $e^{s[\psi(0), H]}$ that leads to a new non-unitary operator $\mathcal{Q}_k^{(2)}$:
\begin{equation}
\mathcal{Q}_k^{(2)} = \mathds{1} + b_{1,k} e^{-i(1-\varphi)\sqrt{s_k}H} + b_{2,k} e^{i\varphi\sqrt{s_k}H}
\end{equation}
where the coefficients $b_{1,k}, b_{2,k}$ are defined in Eq. \eqref{eq:coefficients_2nd_order}. 
The circuit for implementing second-order PDBQITE is illustrated in Fig. \ref{figure:second_order_LCU}, and requires an LCU circuit with two ancillary qubits and their corresponding controlled-unitary operations. Crucially, this circuit has only a constant multiplicative factor in resource cost increase compared to the exponential increase created by the original DBQITE framework. As in the first-order case, the coefficients  $b_{1,k}$ and  $b_{2,k}$ can be estimated from the output state of the last iteration, or from the time-series of the input state when the time steps are constant, as discussed in Appendix \ref{app:higher_order_formulas} and Appendix \ref{sec:estimating_fk_2order}

Our new formulation allows us again to bound the success probability as given in Lemma \ref{lemma:second_order_prob_success}.

\begin{lemma}
\label{lemma:second_order_prob_success}
    The probability of success $P_{\text{succ}}^{(2)}$ for one iteration of the second-order PDBQITE can be lower bounded as:
    \begin{equation}
        P_{\text{succ}}^{(2)} \geq  \frac{1}{\Big(1 + \sqrt{5s_k} + s_k\Big)^2}.
    \end{equation}
\end{lemma}

\begin{proof}
    The proof can be found in Appendix \ref{app:higher_order_formulas}.
\end{proof}

Despite a decrease in the per-step probability of success, our numerical experiments (see Sec. \ref{sec:experiments}) demonstrate that the capability of performing larger imaginary time steps can result in overall savings in total success probability and circuit depth in reaching a specific target state.

\section{Discussion}
\label{sec:discussion}

In this paper, we introduced PDBQITE, a probabilistic imaginary-time evolution algorithm inspired by the double-bracket quantum ITE (DB-QITE) approach of Gluza et al. \cite{gluza2026double}. Whenever a step succeeds, PDBQITE reproduces the corresponding DB-QITE iterate exactly, and therefore inherits its guarantees: the energy is non-increasing and the fidelity with the ground state is non-decreasing at every step. What PDBQITE avoids is the $\mathcal{O}(3^k)$ growth of circuit depth with the number of iterations $k$ that renders DB-QITE impractical on near-term devices. In its place, each iteration is realized by a single, shallow LCU circuit whose depth is linear in the number of iterations, at the price of turning the evolution into a probabilistic one. 

This probabilistic cost is mild in the regime that matters. Although the cumulative success probability decays exponentially in the number of iterations, the per-step success probability is high and, crucially, independent of the system size. This makes it realistic to carry out the handful of steps needed on current devices, a regime that was previously out of reach for DB-QITE.

We provided numerical evidence of its performance on preparing the ground states of small molecules, where it outperforms PITE \cite{kosugi2022imaginary}, to our knowledge the strongest family of near-term probabilistic ITE algorithms to date. We further found that second-order PDBQITE can outperform the first-order variant. Although each step has a lower success probability and requires an extra ancilla and roughly twice the per-step circuit, it reaches a given target accuracy in few enough iterations that its total success probability is higher at comparable overall circuit depth. Finally, we showed that the choice of initial state is crucial for fast convergence to the ground state. In particular, for classical optimization problems a QAOA circuit provides an excellent initial state for ITE; once a QAOA circuit has been prepared, appending a few PDBQITE steps can be more effective than adding further QAOA layers, in several cases pushing the approximation ratio close to 1.

Finally, we introduced several variants of PDBQITE that improve its performance. The first is a second-order PDBQITE algorithm, which permits larger imaginary time-steps at the cost of a single additional ancillary qubit and one extra controlled real-time evolution. We provide numerical evidence that, for a fixed evolution time $\tau$, the resulting reduction in the number of required iterations more than compensates for the lower per-step success probability and the larger per-step circuit, yielding a net gain in the total probability of success. The second variant boosts the per-step success probability through a technique inspired by oblivious amplitude amplification \cite{berry2014exponential} and fixed-point amplitude amplification \cite{yoder2014fixed}; this comes at the price of slower convergence, but distorts the target state only slightly.

Our work opens up many directions for future research. At first, it is important to understand how higher-order approximations perform in practice and at what cost. As we showed in this paper, adding a single ancilla qubit and an extra controlled-real-time evolution for second-order PDBQITE can improve its performance. This indicates that it is important to understand the resource cost of these higher-order approximations and how much they can enhance the performance of PDBQITE. Secondly, it will be important to investigate the effect of noise in our algorithm. Particularly important will be to investigate the effect of shot-noise in the evaluation of the time-series that later define the LCU circuits used in the following step, as this will degrade the quality of the state as the number of steps grows. It will also be beneficial to understand how to appropriately choose the time steps $s_k$, at each iteration. Understanding how to choose the time steps $s_k$ can lead to significantly faster convergence, which requires circuits with a noticeable smaller depth. However, more work is needed to understand how to do this in practice efficiently.

\bibliographystyle{unsrt}
\bibliography{References_1}

\onecolumngrid

\appendix

\section{Comparison with existing methods}
\label{app:comparison_with_other_methods}

An early fault-tolerant family of ITE algorithms builds on \textit{quantum signal processing} and spectral-filtering ideas, preparing a ground state in time polynomial in $\tau$ whenever the initial state has an inverse-polynomial overlap with it \cite{silva2023fragmented,dong2022ground,chan2023simulating,zhang2025quantum,jo2026deterministic}. Reference \cite{zhang2025quantum} sets a particularly strong benchmark with a per-shot acceptance probability that scales as $\Theta(\gamma^2)$, \emph{independently of $\tau$}, at query depth $\widetilde{O}(\tau)$ and with a single ancilla. Yet, even with this substantial improvement over earlier fault-tolerant proposals, its circuits remain far too deep for near-term hardware, demanding long, coherent interleavings of controlled forward and inverse evolution. Moreover, these methods require a good prior estimate of the ground-state energy in order to place the spectral filter.

A second, well-established class of NISQ algorithms for ITE is deterministic, ancilla-free, and uses no controlled evolution. The first lineage, QITE, is tomography-based \cite{motta2020determining,gomes2020efficient,sun2021quantum,huang2023efficient,delcastillo2025multiple}; the second adapts ideas from variational quantum algorithms \cite{mcardle2019variational,yuan2019theory,gomes2021adaptive,gacon2024variational,anuar2024operator, kolotouros2025accelerating}. QITE \cite{motta2020determining} buys determinism at a cost that is exponential in the correlation length of the imaginary-time-evolved state and that further requires solving an ill-conditioned linear system at every step. Variational methods \cite{mcardle2019variational} instead fix the circuit depth by construction, independently of $\tau$, but cap the achievable accuracy at the expressivity of the ansatz, offer no convergence guarantee, and incur a substantial cost in the parameter-optimization loop.

Our proposed algorithm belongs to the family of \textit{double-bracket algorithms} \cite{gluza2026double,zander2025role,suzuki2025double,robbiati2026double,mcmahon2025equating}. The DB-QITE lineage is coherent, ancilla-free and entirely post-selection-free, with a per-step cooling guarantee tying the energy decrease to the current energy variance. However, the biggest obstruction is depth, as the Hamiltonian-simulation calls grow as $O(3^k)$ and only a moderate number of them become realistically reachable. Our proposed algorithm PDBQITE targets depth growth at the cost of adding a probabilistic success that decays exponentially with the number of steps. This is mitigated by its relatively high success rate per step, making it possible to implement a significantly larger number of iterations than previous proposals.

Other near-term probabilistic quantum algorithms for ITE have also been proposed prior to this work \cite{kosugi2022imaginary,nishi2024quadratic,nishi2023optimal,leadbeater2024non,ejima2025probabilistic}. One of the most popular probabilistic ITE methods, named PITE, was proposed by Kosugi et al. in \cite{kosugi2022imaginary}. It requires one ancilla qubit and two controlled real-time evolutions, comparable in cost to our second-order approach. However, PITE achieves weaker results than second-order PDBQITE, and most times weaker than the first-order, in the problems we analyzed (LiH and Heisenberg ground states), a result we believe will generalize to other problems (see Appendix \ref{sec:additional_experiments} for additional experiments). Another potential weakness of PITE is the requirement of an additional hyperparameter $m_0$ that affects the convergence (and robustness) of the algorithm, but also controls the per-step probability of success. The fact that PDBQITE lacks this hyperparameter, together with its better performance, is one of the big advantages of our algorithm over PITE. Other probabilistic ITE methods, such as \cite{xie2024probabilistic, leadbeater2024non}, exploit cheaper gadgets by Trotterizing the imaginary-time evolution and require measuring one ancilla qubit per Hamiltonian term, which can become impractical for Hamiltonians decomposed into a large number of Pauli strings.


Finally, there are other approaches based on LCU gadgets, multi-copies route and a statistical importance sampling approach. LCU gadgets \cite{liu2021probabilistic,xie2024probabilistic,rrapaj2025exact,yi2025probabilistic,kim2026finite} use only a basis change, a controlled $R_y$ and a measurement, giving a shallow per-step circuit; however, these methods require one post-selection per Hamiltonian term. A multi-copy route \cite{schwartzman2026imaginary} completely removes post-selection via a deterministic unitary dilation with unit success probability, paying instead in $2^n$ (tree) or $2n$ (hedge) copies of the full register \cite{alipour2025state}. Finally, if one is interested only to compute expectation values, it is possible to design Monte-Carlo sampling strategies based on averaging over randomly sampled evolutions \cite{huo2023error,martyn2025halving,ray2026quasiprobabilistic,arulandu2026trading,tang2026interference}, which reduce the quantum circuit depth at the price of increasing the number of samples required. 

\section{Extending to higher-order product formulas}
\label{app:higher_order_formulas}

In our analysis, in the main text, we considered the second-order approximation of the operator $e^{s[\psi(0), H]}$ (see Eq. \eqref{eq:first_order_approximation}) as:
\begin{equation*}
    e^{s[\psi(0), H]} = e^{i\sqrt{s}H}e^{i\sqrt{s}\psi(0)}e^{-i\sqrt{s}H}e^{-i\sqrt{s}\psi(0)} + \mathcal{O}(s^{3/2})
\end{equation*}
We can proceed to the next-order approximation \cite{chen2022efficient} and get:
\begin{equation*}
    e^{s[\psi(0), H]} = e^{i\varphi\sqrt{s}H}e^{i\varphi\sqrt{s}\psi(0)}e^{-i\sqrt{s}H}e^{-i(\varphi+1)\sqrt{s}\psi(0)}e^{i(1-\varphi)\sqrt{s}H}e^{i\sqrt{s}\psi(0)} + \mathcal{O}(s^2)
\end{equation*}
where $\varphi = (\sqrt{5}-1)/2$. As a result, assuming that the system is in the state $\ket{\psi_k}$ at iteration $k$, we can acquire the state $\ket{\psi_{k+1}}$ (up to a global phase $e^{i\sqrt{s_k}}$) as:
\begin{equation}
    \ket{\psi_{k+1}} = e^{i\varphi\sqrt{s_k}H}e^{i\varphi\sqrt{s_k}\ketbra{\psi_k}}e^{-i\sqrt{s_k}H}e^{-i(\varphi+1)\sqrt{s_k}\ketbra{\psi_k}}e^{i(1-\varphi)\sqrt{s_k}H}\ket{\psi_k}
\label{eq:ite_evolution_second_order}
\end{equation}
Similarly to our analysis in the main text, we can rewrite Eq. \eqref{eq:ite_evolution_second_order} as:
\begin{equation*}
    \begin{gathered}
        \ket{\psi_{k+1}} = e^{i\varphi\sqrt{s_k}H}e^{i\varphi\sqrt{s_k}\ketbra{\psi_k}}e^{-i\sqrt{s_k}H}e^{-i(\varphi+1)\sqrt{s_k}\ketbra{\psi_k}}e^{i(1-\varphi)\sqrt{s_k}H}\ket{\psi_k} \\
        = e^{i\varphi\sqrt{s_k}H}\left(\mathds{1} + (e^{i\varphi\sqrt{s_k}}-1)\ket{\psi_k}\bra{\psi_k}\right)e^{-i\sqrt{s_k}H}\left(\mathds{1} +(e^{-i(\varphi+1)\sqrt{s_k}}-1)\ket{\psi_k}\bra{\psi_k}\right)e^{i(1-\varphi)\sqrt{s_k}H} \ket{\psi_k}\\
        = e^{i\varphi \sqrt{s_k}H}\Big(e^{-i\sqrt{s_k}H}  + (e^{-i(\varphi+1)\sqrt{s_k}}-1)e^{-i\sqrt{s_k}H}\ketbra{\psi_k} +(e^{i\varphi\sqrt{s_k}}-1)\ketbra{\psi_k}e^{-i\sqrt{s_k}H}\\
        +(e^{i\varphi\sqrt{s_k}}-1)(e^{-i(\varphi+1)\sqrt{s_k}}-1)\ketbra{\psi_k} e^{-i\sqrt{s_k}H}\ketbra{\psi_k}\Big)e^{i(1-\varphi)\sqrt{s_k}H}\ket{\psi_k} \\
        =\Big( e^{-i\sqrt{s_k}(1-\varphi)H} +(e^{-i(\varphi+1)\sqrt{s_k}}-1)e^{-i\sqrt{s_k}(1-\varphi)H}\ketbra{\psi_k} + (e^{i\varphi\sqrt{s_k}}-1)e^{i\varphi\sqrt{s_k}H}\ketbra{\psi_k} e^{-i\sqrt{s_k}H} \\
        +(e^{i\varphi\sqrt{s_k}}-1)(e^{-i(\varphi+1)\sqrt{s_k}}-1)e^{i\varphi\sqrt{s_k}H}\ketbra{\psi_k} e^{-i\sqrt{s_k}H}\ketbra{\psi_k}\Big)e^{i(1-\varphi)\sqrt{s_k}H}\ket{\psi_k}\\
        = (\mathds{1} + b_{1,k} e^{-i(1-\varphi)\sqrt{s_k}H} + b_{2,k} e^{i\varphi\sqrt{s_k}H})\ket{\psi_k}
    \end{gathered}
\end{equation*}
As such:
\begin{equation}
        \mathcal{Q}_k^{(2)} = (\mathds{1} + b_{1,k} e^{-i(1-\varphi)\sqrt{s_k}H} + b_{2,k} e^{i\varphi\sqrt{s_k}H})\ket{\psi_k} 
\label{eq:nonunitary_operator_second_order}
\end{equation}
where
\begin{gather}
    b_{1,k} = (e^{-i(\varphi+1)\sqrt{s_k}}-1)f_k\big((\varphi-1)\sqrt{s_k}\big) \\
    b_{2,k} = \Big(e^{i\varphi\sqrt{s_k}}-1\Big)\Big(f_k\big(\varphi\sqrt{s_k}\big) + (e^{-i(\varphi+1)\sqrt{s_k}}-1)f_k\big(\sqrt{s_k}\big)f_k\big((\varphi-1)\sqrt{s_k}\big)\Big)
\label{eq:coefficients_2nd_order}
\end{gather}
In this case, the LCU circuit will require 2 ancillas and 2 controlled-real-time evolution operators, as seen in Fig. \ref{figure:second_order_LCU}. The $U_{\mathsf{PREP}}$ circuit in this case is different from the single qubit case and is constructed such that:
\begin{equation}
    U_{\mathsf{PREP}}\ket{00} = \frac{1}{\sqrt{\beta_k}}\Big(\ket{00} + \sqrt{|b_{1,k}|}\ket{01} + \sqrt{|b_{2,k}|}\ket{10}\Big)
\end{equation}
with the coefficients $b_{\ell,k}$ given in Eq. \eqref{eq:coefficients_2nd_order}, and $\beta_k = 1+|b_{1,k}|+|b_{2,k}|$. $U_{\mathsf{PREP}}$ can be implemented by first applying $R_y(2\arctan\sqrt{|b_{2,k}|/(1+|b_{1,k}|)}$ on the first qubit, and then a controlled-$R_y(2\arctan\sqrt{|b_{1,k}|})$ with the first qubit as the control (and controlled on being in the $\ket{0}$ state) and the second qubit as the target.

After applying $U_{\mathsf{PREP}}$, we need to apply the correct phases $\phi_1,\phi_2$ corresponding to the angles of the complex coefficients $b_{\ell,k}$. This can be done easily by applying two phase gates:
\begin{equation}
    U_{\text{phase}} \equiv  P_1(\phi_2)\otimes P_2(\phi_1) 
\end{equation}
Next, we need to apply the $\mathsf{SELECT}$ operator which applies a series of controlled operations:
\begin{equation}
    U_{\mathsf{SEL}} = \Big(\ketbra{00} + \ketbra{11}\Big)\otimes \mathds{1} + \ketbra{01}\otimes e^{-i(1-\varphi)\sqrt{s_k}H} + \ketbra{10}\otimes e^{i\varphi\sqrt{s_k}H}
\end{equation}
Finally, we apply $U_{\mathsf{PREP}}^{\dagger}$ and measure in the computational basis. If the ancilla qubits are measured to be in the $\ket{00}$ state, then the post-measurement state is, up to the discretization error of the group commutator formula, proportional to $e^{s_k[\psi_k, H]}|\psi_k\rangle$ (which itself approximates the imaginary-time-evolved state to order $s_k^2$). Thus, we conclude that:
\begin{equation}
(\bra{00}\otimes\mathds{1})U_{\mathsf{PREP}}^{\dagger}U_{\mathsf{SEL}}U_{\text{phase}}U_{\mathsf{PREP}}\ket{00}\ket{\psi_k} = \frac{1}{\beta_k}\mathcal{Q}_k^{(2)}\ket{\psi_k}
\end{equation}
where the above state is subnormalized. Similarly to our previous analysis, we can bound the probability of success for implementing $\mathcal{Q}_k^{(2)}$. First, note that (similar to Lemma \ref{lemma:isometry_on_psi_k}):
\begin{equation}
\norm{\mathcal{Q}_k^{(2)}\ket{\psi_k}} = 1
\end{equation}
Repeating the LCU calculation of Lemma \ref{lemma:prob_of_success_first_order}, we find that the probability of success is:
\begin{equation}
P^{(2)}_{\text{succ}} = \frac{\bra{\psi_k}\bigl(\mathcal{Q}_k^{(2)}\bigr)^\dagger \mathcal{Q}_k^{(2)} \ket{\psi_k}}{\beta^2_k} = \frac{1}{\beta^2_k}.
\end{equation}
Thus, by bounding each coefficient $b_{\ell,k}$, we can get a lower bound on the success probability. Starting from $b_{1,k}$, we have
\begin{equation*}
    b_{1,k} = \bigl(e^{-i(\varphi+1)\sqrt{s_k}}-1\bigr)f_k\bigl((\varphi-1)\sqrt{s_k}\bigr).
\end{equation*}
Using $|e^{ix}-1| = 2|\sin(x/2)| \leq |x|$ for $x \in \mathbb{R}$, and $|f_k(t)| = |\bra{\psi_k}e^{-iHt}\ket{\psi_k}| \leq 1$, we get:
\begin{equation}
|b_{1,k}| \leq (\varphi+1)\sqrt{s_k}.
\label{eq:a1_bound}
\end{equation}
Similarly, for $b_{2,k}$ we have:
\begin{equation*}
  b_{2,k} = \bigl(e^{i\varphi\sqrt{s_k}}-1\bigr)\Big[f_k(\varphi\sqrt{s_k}) + \bigl(e^{-i(\varphi+1)\sqrt{s_k}}-1\bigr)f_k(\sqrt{s_k})f_k\bigl((\varphi-1)\sqrt{s_k}\bigr)\Big].  
\end{equation*}
If we then use the triangle inequality, and the same two bounds as above, we have:
\begin{equation}
|b_{2,k}| \leq \varphi\sqrt{s_k} + \varphi(\varphi+1)s_k.
\end{equation}
Next, we use the fact that $\varphi^{2} = 1 - \varphi$, hence $\varphi(\varphi+1) = \varphi^2 + \varphi = 1$, and thus:
\begin{equation}
|b_{2,k}| \leq\varphi\sqrt{s_k} + s_k.
\label{eq:a2_bound}
\end{equation}
Finally, if we combine the bounds in Eqs. \eqref{eq:a1_bound},\eqref{eq:a2_bound}, we get:
\begin{equation*}
\beta_k \leq 1 + (\varphi+1)\sqrt{s_k} + \varphi\sqrt{s_k} + s_k = 1 + (2\varphi+1)\sqrt{s_k} + s_k =  1 + \sqrt{5s_k} + s_k.
\end{equation*}
Thus, we can conclude that the probability of success for the second-order PDBQITE is lower bounded as:
\begin{equation}
    P_{\text{succ}}^{(2)} = \frac{1}{\beta^2_k} \geq \frac{1}{\Big(1 + \sqrt{5s_k} + s_k\Big)^2}.
\end{equation}

To understand the advantage of second-order PDBQITE method against first-order PDBQITE, we performed numerical experiments to understand how much faster we can converge using the second-order method. As we have already discussed, the ability to take larger imaginary-time steps without sacrificing per-step accuracy reduces the total number of steps required to reach a target ground-state fidelity.

\begin{figure}
    \centering
\includegraphics[width=0.5\linewidth]{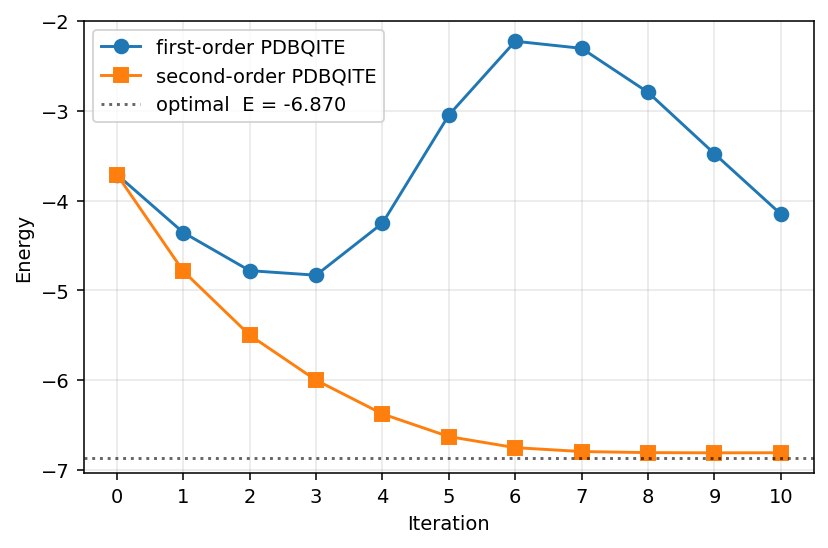}
    \caption{Comparison of first- and second-order PDBQITE for large timesteps $s_k$. For both methods, the system is initialized in the $\ket{+}^{\otimes 10}$ state. Then we apply PDBQITE for 10 steps with a fixed timestep $s_k=0.8$.}
    \label{fig:order1_vs_order2_comparison}
\end{figure}

 In Fig. \ref{fig:order1_vs_order2_comparison}, we plot the energy trajectory for both first- and second-order PDBQITE. The figure corresponds to a MaxCut instance for a 10 node, randomly weighted 3-regular graph. The system is initialized in the $\ket{\psi_0} = \ket{+}^{\otimes 10}$ state, and the timesteps are chosen to be constant with $s=0.8$. As is clearly illustrated, the second-order PDBQITE is necessary if we want to remain close to the actual ITE state, when choosing large timesteps $s$. On the other hand, for the first-order PDBQITE, we can see that after 4 iterations, the energy is increasing, which means that the resulting state is a bad approximation of the ITE state, and thus we need to decrease the timestep size to remain close it.

\section{Estimating $f_k(t)$}
\label{sec:estimating_fk}

As we discussed in previous sections, the implementation of our algorithm requires the estimation of certain expectation values. Specifically, to generate the state $\ket{\psi_{k+1}}$, we need to either apply $\mathcal{Q}_k$ or $\mathcal{Q}_k^{(2)}$ depending on whether we want to use first- or second-order PDBQITE. In the former case, we need to estimate the expectation value $f_k(\sqrt{s_k})$, while in the latter case we need to estimate the expectation values $f_k(\varphi\sqrt{s_k})$, $f_k(\sqrt{s_k})$, and $f_k((\varphi-1)\sqrt{s_k})$ respectively. Recall that $f_k(t)$ is defined as:
\begin{equation}
    f_k(t) := \bra{\psi_k}e^{-iHt}\ket{\psi_k}
\label{eq:overlap}
\end{equation}
We start by estimating $f_0(\sqrt{s}) = \bra{\psi_0}e^{-iH\sqrt{s}}\ket{\psi_0}$ for the initial state $\ket{\psi_0}$. This can be done easily by executing two Hadamard tests as seen in Figure \ref{fig:hadamard_test}. By estimating the probability of measuring $\ket{0}$ on both circuits, we can estimate $f_0(\sqrt{s})$. If $P_R$ is the probability of measuring $\ket{0}$ on the left circuit and $P_I$ is the probability of measuring $\ket{0}$ on the right circuit, then $f_0(\sqrt{s})$ can be estimated as:
\begin{equation}
    f_0(\sqrt{s}) = (2P_R-1) + i (2P_I-1)
\end{equation}

\begin{figure*}
    \begin{quantikz}
\lstick{$\ket{0}$}      & \gate{H}     & \ctrl{1}          & \gate{H}  & \meter{} \\
\lstick{$\ket{\psi_0}$} & \qwbundle{n} & \gate{e^{-iH\sqrt{s}}}  & \qw
\end{quantikz}
\; \;
\begin{quantikz}
    \lstick{$\ket{0}$}& \gate{H}     & \ctrl{1}&\gate{S^{\dagger}}          & \gate{H}  & \meter{} \\
\lstick{$\ket{\psi_0}$} & \qwbundle{n} & \gate{e^{-iH\sqrt{s}}} &\qw& \qw
\end{quantikz}
\caption{Hadamard test to estimate $f_0(\sqrt{s}) = \bra{\psi_0}e^{-iH\sqrt{s}}\ket{\psi_0}$. The circuit on the left is used to estimate the real part of $f_0(\sqrt{s})$, while the circuit on the right is used to estimate its imaginary part.}
\label{fig:hadamard_test}
\end{figure*}
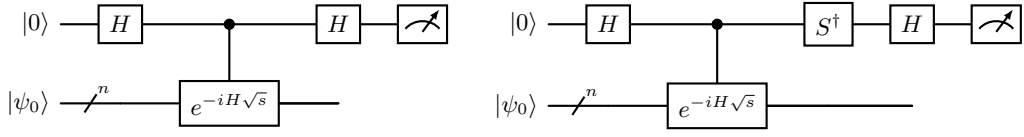

However, estimating $f_k$ for different quantum states $\ket{\psi_k}$ throughout the algorithm can become impractical as the number of iterations grows. The reason is that as $k$ grows, the probability of successfully preparing $\ket{\psi_k}$ shrinks exponentially with $k$. In this section, we will first show how one can estimate $f_k(t)$ by only performing measurements on the initial state $\ket{\psi_0}$, but also show that when the timesteps $s_k$ are constant throughout the evolution, then the quantities $f_k$ change only by a small amount. As a result, in practice, $f_k$ can be taken as a constant for a large number of iterations $k$.

\subsection{First-order PDBQITE}
\label{sec:estimating_fk_1order}

We start by showing that one can avoid estimating $f_k(t)$ by measuring $\ket{\psi_k}$ at each iteration. Let $\sum_ks_k=\tau$ be the discretization of the total evolution time $\tau$. We are interested in estimating $f_{k+1}(\sqrt{s_{k+1}})$, given that we have estimated $f_k(\sqrt{s_k})$ at the previous iteration. We have:
\begin{equation}
    f_{k+1}(\sqrt{s_{k+1}}) = \bra{\psi_{k+1}}e^{-iH\sqrt{s_{k+1}}}\ket{\psi_{k+1}}=
    \bra{\psi_k}\mathcal{Q}_k^{\dagger}e^{-iH\sqrt{s_{k+1}}}\mathcal{Q}_k\ket{\psi_k}
\label{eq:overlap-k+1_step}
\end{equation}
where $\mathcal{Q}_k\ket{\psi_k} = \ket{\psi_{k+1}}$ and $\mathcal{Q}_k$ is given by Eq. \eqref{eq:linear_combination_unitary}. Expanding $\mathcal{Q}_k^{\dagger}e^{-iH\sqrt{s_{k+1}}}\mathcal{Q}_k$ gives:
\begin{equation}
\begin{gathered}
    \mathcal{Q}_k^{\dagger}e^{-iH\sqrt{s_{k+1}}}\mathcal{Q}_k = (\mathds{1}+a_{k}^*e^{-iH\sqrt{s_k}})e^{-iH\sqrt{s_{k+1}}}(\mathds{1}+a_{k}e^{iH\sqrt{s_k}})\\
    =(1+|a_{k}|^2)e^{-iH\sqrt{s_{k+1}}} + a_{k}e^{-iH(\sqrt{s_{k+1}}-\sqrt{s_k})}+a_{k}^*e^{-iH(\sqrt{s_{k+1}}+\sqrt{s_k})}
\end{gathered}
\end{equation}
Plugging it into Eq. \eqref{eq:overlap-k+1_step}, and using the fact that $a_{k}=(e^{i\sqrt{s_k}}-1)f_k(\sqrt{s_k})$ results in:
\begin{equation}
\begin{aligned}
    f_{k+1}(\sqrt{s_{k+1}}) &= \left(1 + 4\sin^2\left(\frac{\sqrt{s_k}}{2}\right)|f_k(\sqrt{s_k})|^2\right)f_k(\sqrt{s_{k+1}}) \\
    &+ \left(e^{i\sqrt{s_k}}-1\right)f_k(\sqrt{s_k})f_k\left(\sqrt{s_{k+1}}-\sqrt{s_k}\right)\\
    &+\left(e^{-i\sqrt{s_k}}-1\right)f_k^*(\sqrt{s_k})f_k\left(\sqrt{s_{k+1}}+\sqrt{s_k}\right)
\end{aligned}
\label{eq:overlap_general_equation}
\end{equation}
We will now assume a constant step size $s$ at each iteration.
As we can see from Eq. \eqref{eq:overlap_general_equation}, when the time step is assumed to be constant throughout the evolution, the overlaps $f_k$ follow a recursive relation. This means that we can estimate $f_k$ by measuring $f_0(t)$ using different timesteps $t$.

For example, suppose that we want to perform 3 iterations of PDBQITE. At the start (step 0), we need to estimate $f_0(\sqrt{s})$. In step 1, we need $f_0(\sqrt{s})$ and $f_0(2\sqrt{s})$. Using these two quantities, we can estimate $f_1(\sqrt{s})$ using Eq. \eqref{eq:overlap_general_equation}. Then in step 2, in order to estimate $f_2(\sqrt{s})$, we need $f_1(\sqrt{s})$, which we calculated in the previous iteration, but also $f_1(2\sqrt{s})$. To estimate the latter, we need $f_0(2\sqrt{s})$ that we estimated in Step 1, but also $f_0(3\sqrt{s})$. Finally, we continue with step 3 where we need $f_3(\sqrt{s})$. To estimate it, we need $f_2(\sqrt{s})$, which we calculated in step 2, but also $f_2(2\sqrt{s})$. To estimate the latter, we need $f_1(2\sqrt{s})$ which we calculated in the previous step, but also $f_1(3\sqrt{s})$ which can be calculated by measuring $f_0(3\sqrt{s})$ and $f_0(4\sqrt{s})$. Thus, we can conclude that at each iteration $k$ of PDBQITE, we need to estimate one additional expectation value:
\begin{equation}
\bra{\psi_0}e^{-iH(k+1)\sqrt{s}}\ket{\psi_0}
\end{equation}
as well as prior information to estimate $f_{k+1}(\sqrt{s})$ using Eq. \eqref{eq:overlap_general_equation}. Note, however, that this becomes harder to calculate for non-commuting Hamiltonians, as the total time-step grows linearly and as a result, more Trotter steps are required to estimate this quantity with the same accuracy. On the other hand, for classical optimization problems such as MaxCut, this does not pose a problem, as the time is just the input on the $R_z$ rotations, provided that we can apply any angle on the available hardware. This implies that the total running cost is two additional Hadamard tests per step, with $t=(k+1)\sqrt{s}$. At this point, we should note that if one does not choose a constant timestep then Eq. \eqref{eq:overlap_general_equation} is impractical as the number of expectation values grows exponentially with the number of iterations. As such, our prior analysis remains efficient as long as we choose a constant timestep throughout the imaginary-time evolution. 
 
We will now show that $f_{k+1}(\sqrt{s})$ remains almost constant over a number of iterations. First of all, we can take advantage of the fact that we need to estimate $f_k(t)$ for small timesteps $t$. Recall that DB-QITE (and thus PDBQITE) are valid approximations of the imaginary-time evolution for small time steps $s_k$. Thus, we can Taylor-expand \eqref{eq:overlap} as:
\begin{equation}
    f_k(t) = 1 -itm_1 - \frac{t^2}{2}m_2 +\frac{it^3}{6}m_3 + \mathcal{O}(t^4)
\label{eq:taylor_expansion_moments}
\end{equation}
where $m_j$ is the $j$-th moment of the Hamiltonian $H$ defined as:
\begin{equation}
    m_j := \bra{\psi_k}H^j\ket{\psi_k}
\end{equation}
As such, we can estimate $f_k(t)$ up to the user-specified precision by estimating moments of $H$. For example, the user can estimate both $\bra{\psi_k}H\ket{\psi_k}$ and $\bra{\psi_k}H^2\ket{\psi_k}$, and then using these expectation values, feed the appropriate value of $t$ and estimate $f_k(t)$. However, the problem is that one would still have to measure moments of $H$ on the state $\ket{\psi_k}$. As we will now show, this can be avoided.

This is very useful in practice, since the user can use the same value $f_k$ for a fixed number of iterations. We are interested in estimating $f_{k+1}(\sqrt{s}) - f_k(\sqrt{s})$. Using Eq. \eqref{eq:overlap_general_equation}, we have:
\begin{equation}
    f_{k+1}(\sqrt{s}) - f_k(\sqrt{s}) = |a_{k}|^2f_k(\sqrt{s})
    + a_{k}
+a_{k}^*f_k\left(2\sqrt{s}\right)
\end{equation}
Using the Taylor expansion of $f_k$ in Eq. \eqref{eq:taylor_expansion_moments}, and $t:=\sqrt{s}$ for easiness,
we can expand $a_{k}$ as:
\begin{equation}
\begin{gathered}
    A_k= \left(it -\frac{t^2}{2}-\frac{it^3}{6}\right)\left(1 -itm_1 - \frac{t^2}{2}m_2 +\mathcal{O}(t^3)\right) \\
    =it +t^2\left(m_1-\frac{1}{2}\right) + it^3\left(\frac{m_1-m_2}{2}-\frac{1}{6}\right) + \mathcal{O}(t^4)
\end{gathered}
\end{equation}
Next, we calculate $|a_{k}|^2$ as:
\begin{equation}
    a_{k}^*A_k= t^2 + \mathcal{O}(t^4)
\end{equation}
and thus:
\begin{equation}
    |a_{k}|^2f_k(t) = t^2 - it^3m_1 +\mathcal{O}(t^4)
\end{equation}
Then, for the next term, we have:
\begin{equation}
    a_{k}^*f_k(2t) = -it -t^2\left(m_1+\frac{1}{2}\right) + it^3\left(\frac{5m_2}{2} + \frac{m_1}{2} -2m_1^2 + \frac{1}{6}\right) + \mathcal{O}(t^4)
\end{equation}
As such if we put everything together, we have:
\begin{equation}
\begin{gathered}
  f_{k+1}(t) - f_k(t) =\\  t^2-it^3m_1 +it +t^2\left(m_1-\frac{1}{2}\right) + it^3\left(\frac{m_1-m_2}{2}-\frac{1}{6}\right) -it -t^2\left(m_1+\frac{1}{2}\right) + it^3\left(\frac{5m_2}{2} + \frac{m_1}{2} -2m_1^2 + \frac{1}{6}\right) =\\
  2i(m_2-m_1^2)t^3 + \mathcal{O}(t^4)
\end{gathered}
\end{equation}
Now, if we use the fact that the variance $\sigma_k^2$ is $\sigma_k^2 = m_2-m_1^2 = \bra{\psi_k}H^2\ket{\psi_k} - (\bra{\psi_k}H\ket{\psi_k})^2$ and that $t=\sqrt{s}$, we have the expression:
\begin{equation}
    f_{k+1}(\sqrt{s}) - f_k(\sqrt{s}) = 2i\sigma_k^2s^{3/2} + \mathcal{O}(s^2)
\end{equation}
which shows the fact that the overlaps differ by an amount of order $\mathcal{O}(s^{3/2})$.

\subsection{Second-order PDBQITE}
\label{sec:estimating_fk_2order}

We can proceed and do the same analysis for second-order PDBQITE; however the calculation can become slightly more tedious that the first-order case. In comparison to the first-order, the second-order non-unitary operator $\mathcal{Q}_k^{(2)}$ is decomposed into a linear combination of three unitary operators, as seen in Eq. \eqref{eq:nonunitary_operator_second_order}.

We are interested in:
\begin{equation}
    f_{k+1}(\tilde{s}) = \bra{\psi_{k+1}}e^{-iH\tilde{s}}\ket{\psi_{k+1}} = \bra{\psi_{k}}\mathcal{Q}_k^{(2)\dagger}e^{-iH\tilde{s}}\mathcal{Q}_k^{(2)}\ket{\psi_{k}} 
\end{equation}
for some $\tilde{s}\in \mathbb{R}$. By expanding the product $\mathcal{Q}_k^{(2)\dagger}e^{-iH\tilde{s}}\mathcal{Q}_k^{(2)}$ we can write it as:
\begin{equation}
\begin{aligned}
    \mathcal{Q}_k^{(2)\dagger}e^{-iH\tilde{s}}\mathcal{Q}_k^{(2)} &= \Big(1 + |b_{1,k}|^2 + |b_{2,k}|^2\Big) e^{-iH\tilde{s}}\\
    &+ b_{1,k}e^{-iH(\tilde{s}+t_1)} + b_{1,k}^*e^{-iH(\tilde{s}-t_1)}\\
    &+ b_{2,k}e^{-iH(\tilde{s}+t_2)} + b_{2,k}^*e^{-iH(\tilde{s}-t_2)}\\
    &+ b_{1,k}^*b_{2,k}e^{-iH(\tilde{s}-t_1+t_2)} + b_{1,k}b_{2,k}^*e^{-iH(\tilde{s}+t_1-t_2)}
\end{aligned}
\end{equation}
where $t_1 = (1-\varphi)\sqrt{s}$, $t_2 = -\varphi \sqrt{s}$ and $t_1-t_2 = \sqrt{s}$. Thus, we can write $f_{k+1}(\tilde{s})$ as:
\begin{equation}
\begin{aligned}
\mathcal{Q}_k^{(2)\dagger}e^{-iH\tilde{s}}\mathcal{Q}_k^{(2)} &= \Big(1 + |b_{1,k}|^2 + |b_{2,k}|^2\Big) f_k(\tilde{s})\\
&+ b_{1,k}f_k(\tilde{s}+(1-\varphi)\sqrt{s}) + b_{1,k}^*f_k(\tilde{s}-(1-\varphi)\sqrt{s})\\
&+ b_{2,k}f_k(\tilde{s}-\varphi \sqrt{s}) + b_{2,k}^*f_k(\tilde{s}-\varphi\sqrt{s})\\
&+ b_{1,k}^*b_{2,k}f_k(\tilde{s}-\sqrt{s}) + b_{1,k}b_{2,k}^*f_k(\tilde{s}+\sqrt{s}) 
\end{aligned}
\end{equation}
We can see that in order to estimate $f_{k+1}(\tilde{s})$ we need to shift $\tilde{s}$ by seven different values and estimate $f_k$ on these values. The seven different ``shifts'' are $\{0, \pm(1-\varphi)\sqrt{s}, \pm \varphi\sqrt{s}, \pm \sqrt{s}\}$. However, the $\pm \sqrt{s},0$ shifts can be written as a linear combination of $(1-\varphi)\sqrt{s}$ and $\varphi\sqrt{s}$ for some $m_1,m_2\in \mathbb{Z}$. Thus, each step of the recurrence adds a vector $(m_1,m_2)$ drawn from the set:
\begin{equation}
    V = \{(0,0), \pm (1,0), \pm(0,1), (1,1)\}
\end{equation}
Every single recurrence step, can add $m_1,m_2\in \{-1, 0, 1\}$. Thus, after $K$ iterations, we will have that $|m_1|\leq d$ and $|m_2|\leq d$. As such, the number of possible integers and therefore different expectation values $f_k$ will be as large as all integers that are within an integer square $(2K+1)\times (2K+1)$. Hence, the number of distinct expectation values needed to calculate $f_k(\sqrt{s})$ for $K$ iterations of second-order PDBQITE scale as $\mathcal{O}(K^2)$.

\section{Amplitude Amplification}
\label{app:amplitude_amplification}

\subsection{Summary}
\label{sec:amplitude_amplification}

As we discuss below, the operator $\mathcal{Q}_k$ is close to a unitary operator, up to an error $\mathcal{O}(\sqrt{s_k})$. When this operator acts on the state $\ket{\psi_k}$, it succeeds with a probability that is large and close to 1 (see Lemma \ref{lemma:prob_of_success_first_order}). This motivates a robust (partial) oblivious amplitude amplification technique, inspired by the work in \cite{berry2014exponential},
that can boost the per-step probability of success to even larger values.
Our construction is based on the walk operator:
\begin{equation}
    \mathcal{W}(\phi) = \mathcal{V}_kR(\phi)\mathcal{V}_k^{\dagger}R(\phi)
\label{eq:partial_walk_operator}
\end{equation}
where $\mathcal{V}_k$ is the LCU unitary that applies the operator $\mathcal{Q}_k$, and $R(\phi)$ is a partial reflection about the ``good'' subspace, i.e., the subspace where the ancilla qubit is in the $\ket{0}$ state, defined as:
\begin{equation}
    R(\phi) \equiv e^{i\phi \Pi_G} = \mathds{1} + (e^{i\phi}-1)\Pi_G
\end{equation}
As we rigorously show in Appendix \ref{app:amplitude_amplification}, applying the walk operator at the end of each PDBQITE iteration can boost the probability of success as
\begin{equation}
    P_{\text{succ}}(\phi)\approx \frac{1}{\alpha_k^2} + \frac{4|a_{k}|}{\alpha_k^4}\phi^2
\end{equation}
when the angle $\phi$ is chosen to be small. 
However, this comes with two costs. First of all, applying one iteration of the walk operator at each iteration triples the depth of the PDBQITE circuit since we have to apply two additional $\mathcal{V}_k$ unitaries. Secondly, due to the non-unitarity of the operator $\mathcal{Q}_k$, boosting the probability of success comes at the cost of distorting the state, as
shown in Fig. \ref{fig:prob_of_success_infidelity} for a 14-qubit MaxCut Hamiltonian corresponding to a randomly weighted 3-regular graph. Note that this distortion could potentially lead to a slower convergence onto the ground state, thus needing more iterations and, as such, deeper quantum circuits.

\begin{figure}[h]
    \centering
\includegraphics[width=1\linewidth]{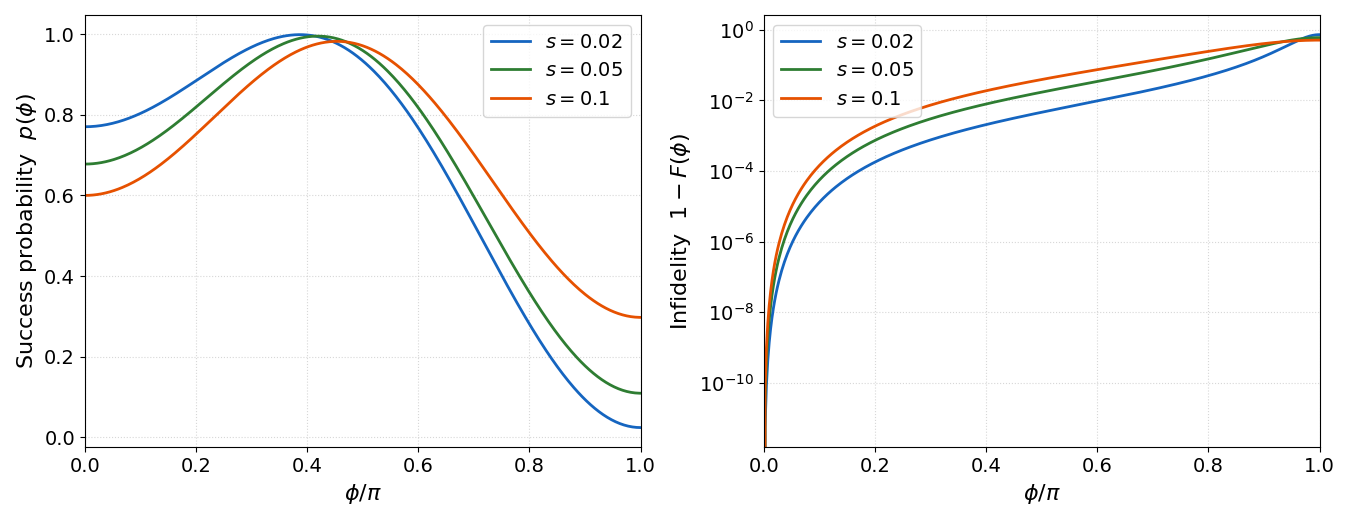}
    \caption{Probability of success (left figure) and infidelity (right figure) as a function of the angle $\phi$ of the walk operator in Eq. \eqref{eq:partial_walk_operator} for a 14-qubit MaxCut Hamiltonian corresponding to a randomly weighted 3-regular graph.. Each line corresponds to a different timestep $s_k$ for one step of PDBQITE.}
\label{fig:prob_of_success_infidelity}
\end{figure}

\subsection{Details}

In this subsection, we will explain how we can amplify the probability of success in PDBQITE algorithm with the cost of a slight distortion in the amplified state. As we see in Eq. \eqref{eq:probability_of_success}, the probability of success is large when we choose a small timestep $s_k$, but not necessarily close to 1. However, we can boost this probability closer to 1 by using a flavor of Amplitude Amplification \cite{brassard2000quantum} that is inspired by \emph{Oblivious Amplitude Amplification} \cite{berry2014exponential}. In the section below, we focus on the first-order approximation of PDB-QITE, in which the linear combination of unitaries is given by Eq. \eqref{eq:linear_combination_unitary}. 

Let $\mathcal{V}$ be the LCU unitary defined in Sec. \ref{sec:pdbqite}, acting on the state $\ket{0}\ket{\psi_k}$ as:
\begin{equation}
\begin{aligned}
    \mathcal{V}_k\ket{0}\ket{\psi_k} &= \frac{1}{\alpha_k}\ket{0}\mathcal{Q}_k\ket{\psi_k} + \sqrt{1-\frac{1}{\alpha_k^2}}\ket{\Phi^{\perp}} \\
    &= \sin q\ket{0}\mathcal{Q}_k \ket{\psi_k} + \cos q \ket{\Phi^{\perp}}
\end{aligned}
 \end{equation}
where $q = \arcsin \sqrt{P_{\text{succ}}}$, and $\ket{\Phi^{\perp}}$ is a $(n+1)$-qubit state that depends on $\ket{\psi_k}$. The state $\ket{\Phi^{\perp}}$ satisfies $\Pi_G\ket{\Phi^{\perp}} = 0$ where $\Pi_G = \ketbra{0}\otimes \mathds{1}$ is the projector onto the ``good'' subspace, i.e., the subspace where the ancilla qubit is in the $\ket{0}$ state.

In the original amplitude amplification of Brassard et al. \cite{brassard2000quantum}, the standard walk operator $\mathcal{W}$ that amplifies the probability of success is defined as:
\begin{equation}
\begin{aligned}
    \mathcal{W} = -\Big(\mathds{1} - 2\mathcal{V}_k\ket{0}\ket{\psi_k}\bra{0}\bra{\psi_k}\mathcal{V}_k^{\dagger}\Big)\Big(\mathds{1} - 2 \Pi_G\Big) \\
    = -\mathcal{V}_k\Big(\mathds{1} - 2\ket{0}\mathcal{Q}_{k-1}\ket{\psi_{k-1}}\bra{0}\bra{\psi_{k-1}}\mathcal{Q}_{k-1}^{\dagger}\Big)\mathcal{V}_k^{\dagger}\Big(\mathds{1} - 2\Pi_G\Big).
\end{aligned}
\end{equation}
However, applying the reflection $\mathds{1} -2\ket{0}Q_{k-1}\ket{\psi_{k-1}}\bra{0}\bra{\psi_{k-1}}\mathcal{Q}_{k-1}^{\dagger}$ requires also amplitude amplification, since the state $\ket{0}\mathcal{Q}_{k-1}\ket{\psi_{k-1}}$ is generated probabilistically, and thus we end up with a recursive relation that takes exponential time. This bottleneck can usually be avoided using \emph{oblivious amplitude amplification} (OAA) \cite{berry2014exponential}. In this case, the walk operator is defined as:
\begin{equation}
\begin{aligned}
    \mathcal{W} = -\Big(\mathds{1} - 2\mathcal{V}_k\ketbra{0}\otimes \mathds{1} \mathcal{V}_k^{\dagger}\Big)\Big(\mathds{1} - 2 \Pi_G\Big) \\
    = -\mathcal{V}_k\Big(\mathds{1} - 2\ketbra{0}\otimes \mathds{1} \Big)\mathcal{V}_k^{\dagger}\Big(\mathds{1}-2 \Pi_G\Big) \\
    = -\mathcal{V}_k\Big(\mathds{1}-2 \Pi_G\Big)\mathcal{V}_k^{\dagger}\Big(\mathds{1}-2 \Pi_G\Big)
\end{aligned}
\end{equation}
and as such, the amplification remains oblivious to the input state. However, according to the 2D Subspace Lemma (see Lemma 3.7 in \cite{berry2014exponential}) OAA can only be used when the linear combination of unitaries is a unitary operator, and this fails in our case. In \cite{berry2015simulating}, the authors extended their result to a robust version of oblivious amplitude amplification called \emph{robust oblivious amplitude amplification} (ROAA), in which the non-unitary operator is close to a unitary operator and the probability of success is close to 1/4. Then, one round of ROAA can boost the probability of success close to one, but with a small error in the amplified state.

Recall that in our case, the operator $\mathcal{Q}_k = \mathds{1} + (e^{i\sqrt{s_k}}-1) f_k(\sqrt{s_k})e^{i\sqrt{s_k}H}$ is close to a unitary operator for a small timestep $s_k$, and the probability of success is close to 1. We take inspiration from fixed-point amplitude amplification \cite{yoder2014fixed} and adapt the idea to the oblivious setting via partial reflections. In this way, we aim to amplify the success probability, but with the cost of introducing a small error in the post-selected state. 

First of all, we need to understand how the timestep $s_k$ affects the non-unitary operator $\mathcal{Q}_k$. Recall that $A_k= (e^{i\sqrt{s_k}}-1) f_k(\sqrt{s_k})$ is the coefficient of the unitary operator $e^{i\sqrt{s_k}H}$ in the LCU (see Eq. \eqref{eq:linear_combination_unitary}), at the $k$-th iteration. We can quantify how non-unitary the operator $\mathcal{Q}_k$ is, using the following Lemma:
\begin{lemma}
    Let $\eta$ be the error that quantifies the non-unitary operator $\mathcal{Q}_k$ as 
    \begin{equation}
        \eta := \norm{\mathcal{Q}_k^\dagger\mathcal{Q}_k - \mathds{1}}.
    \end{equation}
    The error $\eta$ is then upper bounded as:
    \begin{equation}
        \eta = \mathcal{O}(\sqrt{s_k})
    \end{equation}
\end{lemma}

\begin{proof}
    We start by estimating $\mathcal{Q}_k^{\dagger}\mathcal{Q}_k$ as:
    \begin{equation}
        \mathcal{Q}_k^{\dagger}\mathcal{Q}_k = \mathds{1} + A_ke^{iH\sqrt{s_k}} + a_{k}^* e^{-iH\sqrt{s_k}} + |a_{k}|^2\mathds{1}
    \end{equation}
    Then, we have that the error $\eta$ is:
    \begin{equation}
         \\
        \eta = \norm{\mathcal{Q}_k^\dagger\mathcal{Q}_k - \mathds{1}} \leq 2|a_{k}| + |a_{k}|^2 = \mathcal{O}(\sqrt{s_k})
    \end{equation}
    where we used the fact that $|a_{k}| \leq \sqrt{s_k}$ (see proof of Lemma \ref{lemma:prob_of_success_exact_first_order}).
\end{proof}

As we discussed earlier, the probability of success is quite large for small time steps $s_k$. This motivates the use of a partial (and oblivious in our case) reflection, similar to fixed-point amplitude amplification \cite{yoder2014fixed}. We define the reflection $R(\phi)$ that acts only on the ancilla register as:
\begin{equation}
    R(\phi) = e^{i\phi \Pi_G} = \mathds{1} + (e^{i\phi}-1)\Pi_G
\end{equation}
Note that for $\phi=\pi$, we recover the standard (full) reflection used in amplitude amplification ($R(\pi) = \mathds{1}-2\Pi_G$). We can now proceed and introduce an angle-dependent walk operator defined as:
\begin{equation}
    \mathcal{W}(\phi) = \mathcal{V}_kR(\phi)\mathcal{V}_k^{\dagger}R(\phi)
\end{equation}
In the following, we will show how small angles $\phi$ can lead to significant improvements in the success probability, but with the caveat of introducing small errors in the post-selected state. Before our analysis, we need to understand the action of the walk operator $\mathcal{W}(\phi)$ in the quantum state $\mathcal{V}_k\ket{0}\ket{\psi_k}$. We start from $\mathcal{V}_k \ket{0}\ket{\psi_k}$:
\begin{equation}
\begin{gathered}
    \xrightarrow{R(\phi)} \mathcal{V}_k\ket{0}\ket{\psi_k} + \frac{(e^{i\phi}-1)}{\alpha_k}\ket{0}\mathcal{Q}_k\ket{\psi_k} \\
    \xrightarrow{\mathcal{V}_k^{\dagger}} \ket{0}\ket{\psi_k} +  \frac{(e^{i\phi}-1)}{\alpha_k}\mathcal{V}_k^{\dagger}\ket{0}\mathcal{Q}_k\ket{\psi_k}\\
    = \ket{0}\ket{\psi_k} + \frac{(e^{i\phi}-1)}{\alpha_k^2}\ket{0}\mathcal{Q}_k^{\dagger}\mathcal{Q}_k\ket{\psi_k} + \frac{(e^{i\phi}-1)}{\alpha_k}\ket{\Xi^{\perp}}
\end{gathered}
\end{equation}
where the state $\ket{\Xi^{\perp}}$ is defined as:
\begin{equation}
    \ket{\Xi^{\perp}} = \mathcal{V}_k^{\dagger}\ket{0}\mathcal{Q}_k\ket{\psi_k} - \frac{1}{\alpha_k}\ket{0}\mathcal{Q}_k^{\dagger}\mathcal{Q}_k\ket{\psi_k}
\end{equation}
and satisfies $\Pi_G\ket{\Xi^{\perp}} = 0$ (i.e., corresponds to a quantum state where the ancilla qubit is in the state $\ket{1}$). If we continue with the rest of the operations, we get the following:
\begin{equation}
\begin{gathered}
    \xrightarrow{R(\phi)} e^{i\phi}\ket{0}\ket{\psi_k} + \frac{(e^{i\phi}-1)e^{i\phi}}{\alpha_k^2} \ket{0}\mathcal{Q}_k^{\dagger}\mathcal{Q}_k\ket{\psi_k} +\frac{(e^{i\phi}-1)}{\alpha_k}\ket{\Xi^{\perp}}\\
    \xrightarrow{\mathcal{V}_k} \frac{e^{i\phi}}{\alpha_k}\ket{0}\mathcal{Q}_k\ket{\psi_k} + \frac{(e^{i\phi}-1)e^{i\phi}}{\alpha_k^2}\Bigg(\frac{1}{\alpha_k}\ket{0}\mathcal{Q}_k\mathcal{Q}_k^{\dagger}\mathcal{Q}_k\ket{\psi_k} + \ket{\Phi'^{\perp}}\Bigg) \\+ \frac{(e^{i\phi}-1)}{\alpha_k}\ket{0}\mathcal{Q}_k\ket{\psi_k} - \frac{(e^{i\phi}-1)}{\alpha_k^3}\ket{0}\mathcal{Q}_k\mathcal{Q}_k^{\dagger}\mathcal{Q}_k\ket{\psi_k} + \ket{\Xi'^{\perp}}
\end{gathered}
\end{equation}
where we used the fact that:
\begin{equation}
    \mathcal{V}_k\ket{\Xi^{\perp}} = \mathcal{V}_k\Big(\mathcal{V}_k^{\dagger}\ket{0}\mathcal{Q}_k\ket{\psi_k} - \frac{1}{\alpha_k}\ket{0}\mathcal{Q}_k^{\dagger}\mathcal{Q}_k\ket{\psi_k}\Big) = \ket{0}\mathcal{Q}_k\ket{\psi_k} - \frac{1}{\alpha_k^2}\ket{0}\mathcal{Q}_k\mathcal{Q}_k^{\dagger}\mathcal{Q}_k\ket{\psi_k}
\end{equation}
and $\ket{\Phi'^{\perp}}$ and $\ket{\Xi'^{\perp}}$ are states that satisfy $\Pi_G\ket{\Phi'^{\perp}}=0$ and $\Pi_G\ket{\Xi'^{\perp}} = 0$. Acting with $\Pi_G$ results in:
\begin{equation}
    \Pi_G \mathcal{W}(\phi)\mathcal{V}_k\ket{0}\ket{\psi_k} = \frac{1}{\alpha_k}\ket{0}\Big((2e^{i\phi}-1)\mathcal{Q}_k + \frac{(e^{i\phi}-1)^2}{\alpha_k^2}\mathcal{Q}_k\mathcal{Q}_k^{\dagger}\mathcal{Q}_k\Big)\ket{\psi_k}
\label{eq:amplification_partial_reflection}
\end{equation}
Note that for $\phi=0$ (i.e., no walk operator) the second term vanishes, while for $\phi=\pi$ we recover the $\frac{4}{\alpha^3}\mathcal{Q}_k\mathcal{Q}_k^{\dagger}\mathcal{Q}_k$ error of ROAA \cite{berry2015simulating}. Note also that for $\phi=\pi$, in the unitary case, the second term is $\frac{4}{\alpha^3}\mathcal{Q}_k$, which along with the first term result in the amplification of the probability from $\sin q$ to $\sin 3q$.

As we can see from Eq. \eqref{eq:amplification_partial_reflection}, the state after the application of the walk operator and the projector onto the good subspace, results in a state with a leakage due to the fact that the operator $\mathcal{Q}_k$ is non-unitary. In Lemma \ref{lemma:leakage_non_unitarity}, we aim to quantify how large is this residual error, due to the non-unitary nature of $\mathcal{Q}_k$. Intuitively, Eq. \eqref{eq:deviation_from_ideal} aims to quantify the errors that are introduced when $\mathcal{Q}_k$ deviates from being a unitary.

\begin{lemma}
    Let $\mathcal{E}$ be the residual error that quantifies the leakage of the walk operator $\mathcal{W}(\phi)$ due to the non-unitarity of $\mathcal{Q}_k$:
    \begin{equation}
        \mathcal{E} = \norm{\frac{(e^{i\phi}-1)^2}{\alpha_k^3}\mathcal{Q}_k(\mathcal{Q}_k^{\dagger}\mathcal{Q}_k-\mathds{1})\ket{\psi_k}}
    \label{eq:deviation_from_ideal}
    \end{equation}
     The error is then upper bounded as:
    \begin{equation}
        \mathcal{E} = \mathcal{O}(\sqrt{s_k})
    \end{equation}
\label{lemma:leakage_non_unitarity}
\end{lemma}

\begin{proof}
    Consider the product $\mathcal{Q}_k^{\dagger} \mathcal{Q}_k$:
    \begin{equation}
        \mathcal{Q}_k^{\dagger} \mathcal{Q}_k = \mathds{1} + a_{k}e^{iH\sqrt{s_k}} + a_{k}^*e^{-iH\sqrt{s_k}} + |a_{k}|^2\mathds{1} = \mathds{1} + R
    \end{equation}
    where $R\equiv a_{k}e^{iH\sqrt{s_k}} + a_{k}^*e^{-iH\sqrt{s_k}} + |a_{k}|^2\mathds{1}$. On the next step, we need to calculate the squared norm of the vector $R\ket{\psi_k} = (\mathcal{Q}_k^{\dagger}\mathcal{Q}_k-\mathds{1})\ket{\psi_k}$:
    \begin{equation}
        \norm{(\mathcal{Q}_k^{\dagger}\mathcal{Q}_k-\mathds{1})\ket{\psi_k}}^2 = \bra{\psi_k}R^2\ket{\psi_k}
    \end{equation}
    We can express $R^2$ as:
    \begin{equation}
        R^2 = a_{k}^2e^{2iH\sqrt{s_k}} + (a_{k}^*)^2e^{-2iH\sqrt{s_k}} + 2|a_{k}|^2\mathds{1} + 2|a_{k}|^2(a_{k}e^{iH\sqrt{s_k}} + a_{k}^*e^{-iH\sqrt{s_k}}) + |a_{k}|^4\mathds{1}
    \end{equation}
    Taking the expectation value $r^2 = \bra{\psi_k}R^2\ket{\psi_k}$ we have:
    \begin{equation}
        r^2 = 2\Re\left(a_{k}^2f^*_k(2\sqrt{s_k})\right) + 2|a_{k}|^2 + 2|a_{k}|^2(a_{k}f^*_k(\sqrt{s_k}) + a_{k}^*f_k(\sqrt{s_k})) + |a_{k}|^4
    \end{equation}
    Using the result of Lemma \ref{lemma:isometry_on_psi_k}, we see that the parenthesis in the third term is equal to $-|a_{k}|^2$. Thus:
    \begin{equation}
        r^2 = 2\Re \Big(a_{k}^2f_k^*\Big(2\sqrt{s_k}\Big)\Big) + 2|a_{k}|^2 - |a_{k}|^4
    \label{eq:r_squared}
    \end{equation}
    We can then proceed and bound the residual error:
    \begin{equation}
        \norm{(\mathcal{Q}_k^{\dagger}\mathcal{Q}_k-\mathds{1})\ket{\psi_k}}^2 \leq 2|a_{k}|^2 + 2|a_{k}|^2 - |a_{k}|^4 = 4|a_{k}|^2-|a_{k}|^4 \leq 4|a_{k}|^2
    \end{equation}
    We also have $\norm{\mathcal{Q}_k}\leq 1 + |a_{k}| = \alpha_k$ and thus:
    \begin{equation}
        \norm{\mathcal{E}_k}\leq \frac{2|e^{i\phi}-1|^2}{\alpha_k^3}\alpha_k|a_{k}| =\frac{2|e^{i\phi}-1|^2|a_{k}|}{(1+|a_{k}|)^2}
    \end{equation}
    If we then use the fact that $|e^{i\phi}-1|^2 = 4\sin^2(\phi/2)$, we can conclude that the residual error can be upper bounded as:
    \begin{equation}
        \norm{\mathcal{E}_k}\leq 8\sin^2\Big(\frac{\phi}{2}\Big)\sqrt{s_k} + \mathcal{O}(s_k)
    \end{equation}
\end{proof}

We therefore have matched the bound in \cite{berry2015simulating} (for $\phi=\pi$) exactly. At this point, it is crucial to understand two important things. First, we need to quantify how much the probability of success is amplified when we perform one iteration of the partial walk operator, and secondly, how much the infidelity with the ideal state ($\ket{0}\mathcal{Q}_k\ket{\psi_k}$) increases as we vary the angle $\phi$ in the partial reflection. We first introduce Lemma \ref{lemma:probability_of_success_exact} that quantifies how much the probability of success is amplified after one iteration of the partial walk operator. 

\begin{lemma}
    The probability of success (i.e., the probability of measuring the ancilla qubit to be in $\ket{0}$ state) after one iteration of the partial walk operator
    \begin{equation*}
        \mathcal{W}(\phi) = \mathcal{V}_kR(\phi)\mathcal{V}_k^{\dagger}R(\phi)
    \end{equation*}
    on the state $\mathcal{V}_k\ket{0}\ket{\psi_k}$ is:
    \begin{equation}
        P_{\text{succ}}(\phi) = \frac{1+8\sin^2(\phi/2)}{\alpha_k^2} -\frac{8(1+r^2)\sin^2(\phi/2)(1+2\sin^2(\phi/2))}{\alpha_k^4} + \frac{16\sin^4(\phi/2)(1+3r^2+ \langle R^3\rangle)}{\alpha_k^6}
    \label{eq:probability_of_success_exact}
    \end{equation}
    where:
    \begin{gather}
        r^2 =2\Re \Big(a_{k}^2f_k^*\Big(2\sqrt{s_k}\Big)\Big) + 2|a_{k}|^2 - |a_{k}|^4\\
    \langle R^3\rangle = 2\Re \Big(a_{k}^3f_k\Big(3\sqrt{s_k}\Big)\Big) +  6|a_{k}|^2\Re \Big(a_{k}^2f_k^*\Big(2\sqrt{s_k}\Big)\Big) + 3|a_{k}|^4 - 2|a_{k}|^6
    \end{gather}
    and $R$ is defined as $\mathcal{Q}_k^{\dagger}\mathcal{Q}_k = \mathds{1} + R$.
\label{lemma:probability_of_success_exact}
\end{lemma}

\begin{proof}
Recall that the subnormalized state, after measuring the ancilla qubit in the $\ket{0}$ state, is given in Eq. \eqref{eq:amplification_partial_reflection}. Let $\ket{\chi(\phi)} = \Pi_G \mathcal{W}(\phi)\mathcal{V}_k\ket{0}\ket{\psi_k}$:
\begin{equation}
    \ket{\chi(\phi)} = \ket{0}\Bigg(\frac{(2e^{i\phi}-1)}{\alpha_k}\mathcal{Q}_k\ket{\psi_k} + \frac{(e^{i\phi}-1)^2}{\alpha_k^3} \mathcal{Q}_k\mathcal{Q}_k^{\dagger}\mathcal{Q}_k\ket{\psi_k}\Bigg)
\end{equation}
The probability of success is $p(\phi) = \norm{\ket{\chi(\phi)}}^2$. By calculating the squared norm, we have:
\begin{equation}
    \begin{aligned}
        \norm{\ket{\chi(\phi)}}^2  &=\frac{5-4\cos \phi}{\alpha_k^2}\bra{\psi_k}\mathcal{Q}_k^{\dagger}\mathcal{Q}_k\ket{\psi_k} \\
        &+ 2\Re\Bigg(\frac{(2e^{-i\phi}-1)(e^{i\phi}-1)^2}{\alpha_k^4} \bra{\psi_k}\mathcal{Q}_k^{\dagger}\mathcal{Q}_k\mathcal{Q}_k^{\dagger}\mathcal{Q}_k\ket{\psi_k}\Bigg)\\
        &+\frac{|e^{i\phi}-1|^4}{\alpha_k^6}\bra{\psi_k}\mathcal{Q}_k^{\dagger}\mathcal{Q}_k\mathcal{Q}_k^{\dagger}\mathcal{Q}_k\mathcal{Q}_k^{\dagger}\mathcal{Q}_k\ket{\psi_k}
    \end{aligned}
\label{eq:prob_of_success_derivation_1}
\end{equation}
Starting from the first term in Eq. \eqref{eq:prob_of_success_derivation_1}, we have that by using Lemma \ref{lemma:isometry_on_psi_k} and the fact that $\cos\phi = 1-2\sin^2(\phi/2)$:
\begin{equation}
    \frac{5-4\cos \phi}{\alpha_k^2}\bra{\psi_k}\mathcal{Q}_k^{\dagger}\mathcal{Q}_k\ket{\psi_k} = \frac{1+8\sin^2(\phi/2)}{\alpha_k^2}
\end{equation}
Next, we continue with the second term. Note, that the operator $\mathcal{Q}_k^{\dagger}\mathcal{Q}_k$ is hermitian, and as such it has real eigenvalues. We have that:
\begin{equation*}
\bra{\psi_k}\mathcal{Q}_k^{\dagger}\mathcal{Q}_k\mathcal{Q}_k^{\dagger}\mathcal{Q}_k\ket{\psi_k} = \bra{\psi_k}(\mathcal{Q}_k^{\dagger}\mathcal{Q}_k)^2\ket{\psi_k}
\end{equation*}
Recall that we can write $\mathcal{Q}_k^{\dagger}\mathcal{Q}_k$ (see proof of Lemma \ref{lemma:leakage_non_unitarity}) as:
\begin{equation}
    \mathcal{Q}_k^{\dagger}\mathcal{Q}_k = \mathds{1} + R
\end{equation}
where $R = A_ke^{iH\sqrt{s_k}} + a_{k}^* e^{-iH\sqrt{s_k}} + |a_{k}|^2\mathds{1}$. Expressing the operator as this allows us to use the fact that $\bra{\psi_k}R\ket{\psi_k} = 0$ (see Lemma \ref{lemma:isometry_on_psi_k}). Thus, we can conclude that:
\begin{equation}
    \bra{\psi_k}(\mathcal{Q}_k^{\dagger}\mathcal{Q}_k)^2\ket{\psi_k} = 1 +r^2
\end{equation}
where $r^2=\bra{\psi_k}R^2\ket{\psi_k}$. We can then calculate $r^2$ as (see proof of Lemma \ref{lemma:leakage_non_unitarity}):
\begin{equation}
\begin{gathered}
    r^2 = \bra{\psi_k}R^2\ket{\psi_k} 
    = 2\Re \Big(a_{k}^2f_k^*\Big(2\sqrt{s_k}\Big)\Big) + 2|a_{k}|^2 - |a_{k}|^4
\end{gathered}
\end{equation}
Thus, if we use the fact that $2\Re\Big((2e^{-i\phi}-1)(e^{i\phi}-1)^2\Big) = -4\sin^2(\phi/2)(1 + 2\sin^2(\phi/2))$ we can conclude that the second term is written as:
\begin{equation}
     2\Re\Bigg(\frac{(2e^{-i\phi}-1)(e^{i\phi}-1)^2}{\alpha_k^4} \bra{\psi_k}\mathcal{Q}_k^{\dagger}\mathcal{Q}_k\mathcal{Q}_k^{\dagger}\mathcal{Q}_k\ket{\psi_k}\Bigg) = \frac{-8(1+r^2)\sin^2(\phi/2)(1 +2\sin^2(\phi/2))}{\alpha_k^4}
\end{equation}
Next, for the third term we need to calculate the norm $\norm{\mathcal{Q}_k\mathcal{Q}_k^{\dagger}\mathcal{Q}_k\ket{\psi_k}}^2 = \bra{\psi_k}(\mathcal{Q}_k^{\dagger}\mathcal{Q}_k)^3\ket{\psi_k}$. We have:
\begin{equation}
    \bra{\psi_k}(\mathds{1}+R)^3\ket{\psi_k} = 1 + 3r^2 + \bra{\psi_k}R^3\ket{\psi_k}
\label{eq:norm_R_cubed}
\end{equation}
Estimating the last expectation value gives:
\begin{equation}
    \bra{\psi_k}R^3\ket{\psi_k} = 2\Re(a_{k}^3f^*_k(3\sqrt{s_k})) +6|a_{k}|^2\Re(a_{k}^2f_k^*(2\sqrt{s_k})) + 3|a_{k}|^4 -2|a_{k}|^6
\end{equation}
Finally, if we use the fact that $|e^{i\phi}-1|^4 = 16\sin^4(\phi/2)$, we end up with Eq. \eqref{eq:prob_of_success_derivation_1} and thus finish the proof.
\end{proof} 

Although the exact form of the probability of success can look intimidating, the form can be simplified if we restrict the computation to small angles $\phi$. This is expressed in Corollary \ref{corollary:prob_of_success_small_angles}.

\begin{corollary}
    For small angles $\phi$ the probability of success in Eq. \eqref{eq:probability_of_success} can be written as:
    \begin{equation}
        P_{\text{succ}}(\phi)\approx \frac{1}{\alpha_k^2} + \frac{4|a_{k}|}{\alpha_k^4}\phi^2
    \end{equation}
\label{corollary:prob_of_success_small_angles}
\end{corollary}

Since we have an exact expression of the probability of success, we can estimate the optimal angle $\phi^*$ that leads to a probability of success close to 1. However, as we discussed, the non-unitarity of the operator $\mathcal{Q}_k$ introduces leakage that depends on the angle $\phi$. This means that the larger the amplification, the greater the error from the target state $\mathcal{Q}_k\ket{\psi_k}$. Next, we quantify the infidelity of the amplified state with the true success state.

Let $\ket{\psi^{\text{amp}}_k}$ be the quantum state when we measure the ancilla qubit of the amplified state to be in the $\ket{0}$ state:
\begin{equation}
    \ket{\psi^{\text{amp}}_k} := \frac{ \Pi_G \mathcal{W}(\phi)\mathcal{V}_k\ket{0}\ket{\psi_k}}{\norm{ \Pi_G \mathcal{W}(\phi)\mathcal{V}_k\ket{0}\ket{\psi_k}}}
\end{equation}

    Recall that the amplified state in Eq. \eqref{eq:amplification_partial_reflection} is already decomposed into a part which corresponds to the target state $\mathcal{Q}_k\ket{\psi_k}$ and to a part that is generated due to the fact that $\mathcal{Q}_k$ is a non-unitary operator. We can then decompose the latter part into two orthogonal components as:
    \begin{equation}
        \mathcal{Q}_k\mathcal{Q}_k^{\dagger}\mathcal{Q}_k\ket{\psi_k} = \mu \mathcal{Q}_k \ket{\psi_k} + \xi \ket{\perp}
    \end{equation}
    such that $\norm{\ket{\perp}} = 1$ and $\bra{\psi_k}\mathcal{Q}_k^{\dagger}\ket{\perp} = 0$. We first estimate $\mu$ as:
    \begin{equation}
        \mu = \bra{\psi_k}\mathcal{Q}_k^{\dagger}\mathcal{Q}_k\mathcal{Q}_k^{\dagger}\mathcal{Q}_k\ket{\psi_k} = 1 +\bra{\psi_k}R^2\ket{\psi_k} = 1 +2\Re \Big(a_{k}^2f_k^*\Big(2\sqrt{s_k}\Big)\Big) + 2|a_{k}|^2 - |a_{k}|^4
    \end{equation}
    where we used Eq. \eqref{eq:r_squared}. Then, we can estimate $\xi$:
    \begin{equation}
        \xi = \bra{\perp}\mathcal{Q}_k^{\dagger}\mathcal{Q}_k\mathcal{Q}_k^{\dagger}\mathcal{Q}_k\ket{\psi_k}
    \end{equation}
    using the norm $\norm{\mathcal{Q}_k\mathcal{Q}_k^{\dagger}\mathcal{Q}_k\ket{\psi_k}}^2 = \bra{\psi_k}(\mathds{1} + R)^3\ket{\psi_k}$ from Eq. \eqref{eq:norm_R_cubed}. We then have that:
    \begin{equation}
        |\xi|^2 = \bra{\psi_k}(\mathds{1} +R)^3\ket{\psi_k} -|\mu|^2
    \end{equation}
    We can then rewrite the subnormalized success state as:
    \begin{equation}
    \begin{aligned}
        \Pi_G \mathcal{W}(\phi)\mathcal{V}_k\ket{0}\ket{\psi_k} &=  \frac{1}{\alpha_k}\ket{0}\Big((2e^{i\phi}-1)\mathcal{Q}_k + \frac{(e^{i\phi}-1)^2}{\alpha_k^2}\mathcal{Q}_k\mathcal{Q}_k^{\dagger}\mathcal{Q}_k\Big)\ket{\psi_k}\\
        &=\frac{1}{\alpha_k}\ket{0}\Big((2e^{i\phi}-1)  + \frac{(e^{i\phi}-1)^2\mu}{\alpha_k^2}\Big)\mathcal{Q}_k\ket{\psi_k} + \ket{0}\frac{(e^{i\phi}-1)^2\xi}{\alpha_k^3}\ket{\perp}
    \end{aligned}
    \end{equation}
    Since we have decomposed the amplified state into orthogonal components, we can easily calculate the norm as:
    \begin{equation}
        \norm{\Pi_G \mathcal{W}(\phi)\mathcal{V}_k\ket{0}\ket{\psi_k}}^2 = \frac{|(2e^{i\phi}-1)  + \frac{(e^{i\phi}-1)^2\mu}{\alpha_k^2}|^2 + |\frac{(e^{i\phi}-1)^2\xi}{\alpha_k^2}|^2}{\alpha_k^2}
    \end{equation}
    and thus the fidelity is:
    \begin{equation}
        \mathcal{F}(\phi) =   \frac{\Big|(2e^{i\phi}-1) + \frac{(e^{i\phi}-1)^2\mu}{\alpha_k^2}\Big|^2}{\Big|(2e^{i\phi}-1)  + \frac{(e^{i\phi}-1)^2\mu}{\alpha_k^2}\Big|^2 + \Big|\frac{(e^{i\phi}-1)^2\xi^2}{\alpha_k^2}\Big|^2}
    \end{equation}

\begin{table*}[t]
\centering
\begin{tabular}{l cc cc cc}
\toprule
& \multicolumn{2}{c}{$s=0.05$} & \multicolumn{2}{c}{$s=0.1$} & \multicolumn{2}{c}{$s=0.15$}\\
\cmidrule(lr){2-3}\cmidrule(lr){4-5}\cmidrule(lr){6-7}
Method & $\bar P_{\mathrm{succ}}$ & Iters & $\bar P_{\mathrm{succ}}$ & Iters & $\bar P_{\mathrm{succ}}$ & Iters\\
\midrule
\multicolumn{7}{l}{\textit{(a)} $n=12$ qubits}\\
\addlinespace[1pt]
no AA & $9.28\times10^{-4}$ & 21 & $6.12\times10^{-3}$ & 11 & $1.07\times10^{-2}$ & 8\\
ROAA $\pi/12$ & $1.11\times10^{-3}$ & 23 & $6.48\times10^{-3}$ & 12 & $1.35\times10^{-2}$ & 8\\
ROAA $\pi/6$ & $1.88\times10^{-3}$ & 30 & $1.07\times10^{-2}$ & 15 & $2.07\times10^{-2}$ & 10\\
ROAA $\pi/5$ & $2.83\times10^{-3}$ & 35 & $1.35\times10^{-2}$ & 17 & $2.74\times10^{-2}$ & 11\\
ROAA $\pi/4$ & $6.04\times10^{-3}$ & 46 & $2.41\times10^{-2}$ & 20 & $4.33\times10^{-2}$ & 13\\
\midrule
\multicolumn{7}{l}{\textit{(b)} $n=14$ qubits}\\
\addlinespace[1pt]
no AA & $3.14\times10^{-4}$ & 23 & $2.82\times10^{-3}$ & 12 & $6.15\times10^{-3}$ & 9\\
ROAA $\pi/12$ & $3.73\times10^{-4}$ & 25 & $3.05\times10^{-3}$ & 13 & $7.74\times10^{-3}$ & 9\\
ROAA $\pi/6$ & $6.90\times10^{-4}$ & 33 & $5.10\times10^{-3}$ & 16 & $1.19\times10^{-2}$ & 11\\
ROAA $\pi/5$ & $1.09\times10^{-3}$ & 38 & $7.07\times10^{-3}$ & 18 & $1.55\times10^{-2}$ & 12\\
ROAA $\pi/4$ & $2.60\times10^{-3}$ & 50 & $1.30\times10^{-2}$ & 23 & $2.59\times10^{-2}$ & 15\\
\bottomrule
\end{tabular}
\caption{Cumulative success probability and iteration count to reach an approximation ratio of $r=0.95$ on weighted random 3-regular MaxCut instances, when starting from a QAOA generated state (with parameters chosen from the Dweight method), averaged over 20 instances. $\bar P_{\mathrm{succ}}$ corresponds to the mean cumulative success probability while iterations indicate the number of PDBQITE steps required to reach $r=0.95$.}
\label{tab:aa_comparison}
\end{table*}

We tested different angles $\phi = \{\pi/12, \pi/6, \pi/5, \pi/4\}$ for the amplitude amplification and compared the cumulative probability of success and the number of iterations required to reach an approximation ratio of $r=0.95$. We tested 12 and 14 qubit MaxCut instances (20 of each size), corresponding to randomly weighted 3-regular graphs. We examined how many iterations are needed to reach the target approximation ratio and at the same time what is the cumulative probability of success. The results are illustrated in Table \ref{tab:aa_comparison}. Overall, we can see the practical usefulness of amplitude amplification reaching the same approximation ratio with improved probability of success. However, due to the non-unitarity of $\mathcal{Q}_k$, more iterations are required as the state is distorted after every iteration. Specifically, one iteration of amplitude amplification is 3 times more costly than a single step of PDBQITE without amplitude amplification. Thus, there is a trade-off between a higher probability of success and a number of iterations to reach a target approximation ratio.

\section{Cost of controlled real-time evolution}
\label{app:controlled_real_time_evolution}

In this section, we will outline the quantum resources needed to perform a controlled real-time evolution. We will first discuss diagonal Hamiltonians, e.g. Hamiltonians corresponding to classical combinatorial optimization problems. Without loss of generality, we will use the MaxCut Hamiltonian as the problem Hamiltonian. 

Recall that for a graph $G=(V,E)$, where $V$ is the set of nodes and $E$ is the set of edges, the MaxCut Hamiltonian is written as:
\begin{equation}
    H_{\text{MC}} = \frac{1}{2} \sum_{(i,j)\in E}w_{ij}(\mathds{1} - Z_iZ_j)
\end{equation}
where $w_{ij}$ is the weight of the edge $(i,j)$. Since the Hamiltonian is comprised of commuting terms, and if we neglect the constant part of the Hamiltonian, the real-time evolution can be written as:
\begin{equation}
    e^{-iH_{\text{MC}}t} = \prod_{(i,j)\in E} e^{-iZ_iZ_jt}
\end{equation}
Each term $e^{-iZ_iZ_j t}$ is then decomposed as:
\begin{equation}
    e^{-iZ_iZ_j t} = \mathsf{CNOT}_{ij}R_z(2t)\mathsf{CNOT}_{ij}
\end{equation}
As such, every QAOA layer consists of $2|E|$ $\mathsf{CNOT}$ gates and $|E|$ $R_z$ single-qubit rotations. On the other hand, in PDBQITE, we need to apply controlled real-time evolutions:
\begin{equation}
    \ketbra{0}\otimes \mathds{1} + \ketbra{1}\otimes \mathsf{CNOT}_{ij}R_z(2t)\mathsf{CNOT}_{ij}
\end{equation}
Observe that:
\begin{equation}
    \begin{gathered}
    (\mathds{1}\otimes \mathsf{CNOT}_{ij})(CR_z(2t))(\mathds{1}\otimes \mathsf{CNOT}_{ij})=\\
    \ketbra{0}\otimes \mathds{1} + \ketbra{1} \otimes (\mathsf{CNOT}_{ij}R_z(2t)\mathsf{CNOT}_{ij})
\end{gathered}
\end{equation}
As such, the controlled evolution replaces the $R_z$ rotation by a controlled (on the ancilla qubit) $R_z$ gate. The circuit for a single controlled $e^{-iZ_iZ_jt}$ is depicted in Fig. \ref{fig:contolled-real-time-evolution}. We can then conclude that a single controlled-real-time evolution is as costly as two real-time evolutions for the MaxCut problem.

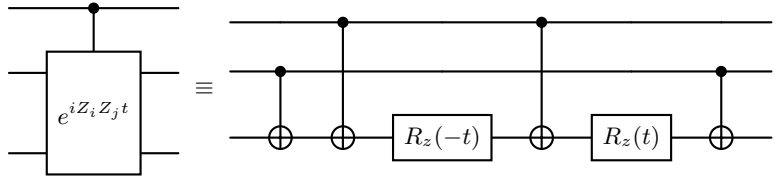
\begin{figure*}
\begin{center}
\begin{quantikz}
    &\ctrl{1}&\\
     &\gate[2]{e^{iZ_iZ_jt}} &\\
     & &
\end{quantikz}
$\equiv$
\begin{quantikz}
    & &\ctrl{2} & &\ctrl{2} &&& \\
     &\ctrl{1} && & & &\ctrl{1}&\\
     &\targ{} &\targ{} &\gate{R_z(-t)} &\targ{} & \gate{R_z(t)} &\targ{}&
\end{quantikz}
\end{center}
\caption{Exact decomposition of the controlled evolution for a single Hamiltonian term $Z_iZ_j$ of the MaxCut Hamiltonian. Each term requires 4 $\mathsf{CNOT}$ gates and two $R_z$ rotations.}
\label{fig:contolled-real-time-evolution}
\end{figure*}

Our analysis can be generalized to commuting Hamiltonians. Specifically, assume that we need to construct the controlled real-time evolution for a Hamiltonian $H=\sum_{\ell}c_{\ell}P_{\ell}$ where $P_{\ell}$ is a Pauli string such that $[P_k, P_l]=0, \forall k,l$. To do so, for each $R_z$ gate in the circuit, we need to add two additional $R_z$ gates and two additional $\mathsf{CNOT}$ gates. In the case of noncommuting Hamiltonians, see \cite{simon2025halving} about the total resources required, but also how to reduce the total cost by half.

\section{Additional Experiments}
\label{sec:additional_experiments}

\subsection{MaxCut}

In this section, we provide additional experiments on the MaxCut problem. Similar to our experiments in Sec. \ref{sec:experiments}, instead of initializing the system in the uniform superposition state $\ket{+}^{\otimes n}$, we choose the quantum state generated by a QAOA \cite{farhi2014quantum} circuit with parameters $(\boldsymbol{\beta}, \boldsymbol{\gamma})$ chosen according to fixed-angle conjecture \cite{vcepaite2025quantum, shaydulin2023parameter, wurtz2021fixed, sureshbabu2024parameter}:
\begin{equation}
    \ket{\psi_0} = U_{\text{QAOA}}^p(\boldsymbol{\beta}, \boldsymbol{\gamma})\ket{+}^{\otimes n}.
\end{equation}
For our benchmarks, we evaluated PDBQITE on randomly weighted 3-regular and 4-regular graphs. The weights of the graphs were sampled from the uniform distribution $\mathcal{U}(0,1)$ and rounded to 3 decimal points.

We tested 20 graph instances for each node size $n\in \{12,14,16\}$ and for up to 14 QAOA layers. Regarding the initial state of our algorithm, we used a QAOA circuit with a varying number of layers $p$, and parameters chosen according to the SKatan method \cite{vcepaite2025quantum, shaydulin2023parameter}. For ITE, we chose both first- and second-order PDBQITE and performed 3 steps. For first-order PDBQITE, we used a timestep $s=0.35$, while for second-order PDBQITE we used a larger timestep, set to $s=0.5$. We argue that one could optimize the timestep choices. However, in our case, we would like to show that even if the timesteps are naively chosen, the method can still perform very well in practice.

\begin{figure*}[t]
    \centering
    \begin{subfigure}[t]{0.48\textwidth}
        \centering
        \caption{3-regular graphs}
        \label{fig:qaoa-vs-p-3reg}
        \includegraphics[width=\linewidth]{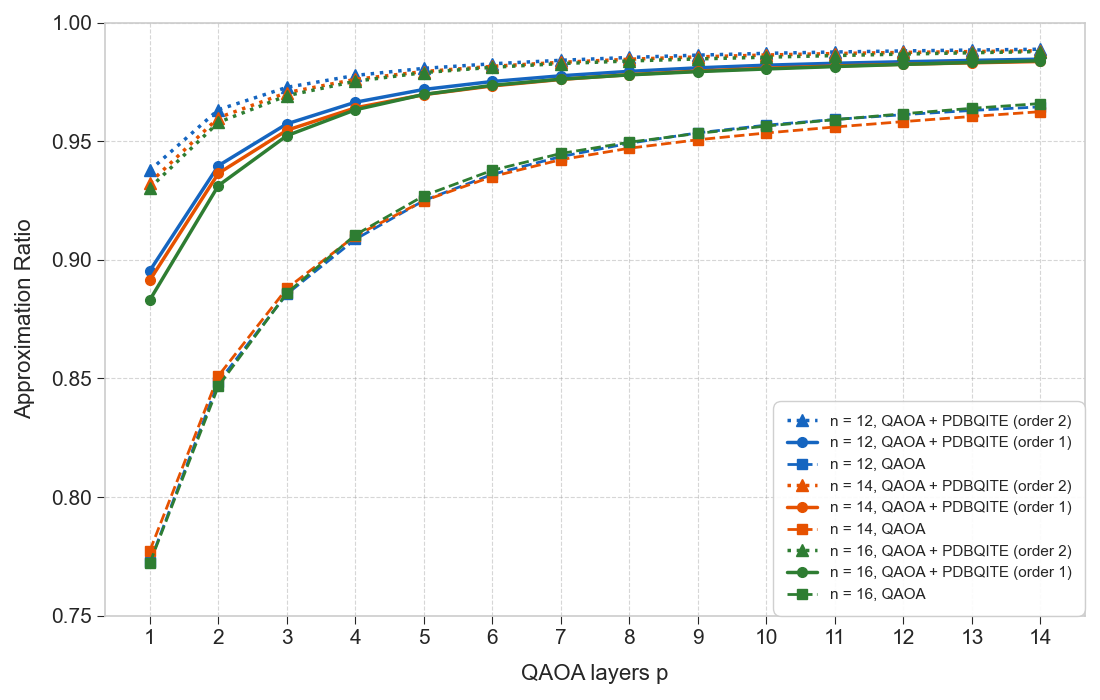}
    \end{subfigure}
    \hfill
    \begin{subfigure}[t]{0.48\textwidth}
        \centering
        \caption{4-regular graphs}
        \label{fig:qaoa-vs-p-4reg}    \includegraphics[width=\linewidth]{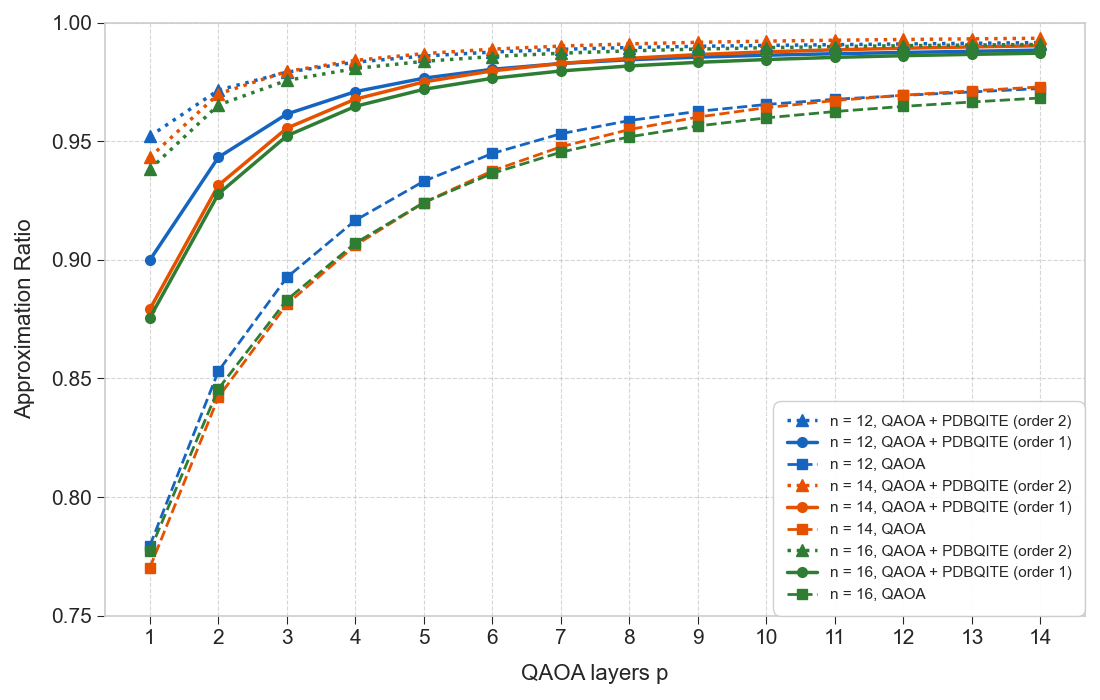}
    \end{subfigure}
    \caption{Approximation ratio as a function of QAOA depth $p$ for randomly weighted d-regular MaxCut
    instances, comparing QAOA against QAOA augmented with
    three PDBQITE step ($s=0.3$ for first-order PDBQITE, while $s=0.5$ for second-order PDBQITE).}
    \label{fig:qaoa_vs_qaoaite_2}
\end{figure*}

The results of our experiments are illustrated in Fig. \ref{fig:qaoa_vs_qaoaite_2}. As we explain in Appendix \ref{app:controlled_real_time_evolution}, one ITE step for first- and second-order PDBQITE is equivalent (in terms of number of single- and two-qubit gates) to two and four QAOA layers respectively. As illustrated in Fig. \ref{fig:qaoa_vs_qaoaite_2}, QAOA is a very powerful technique by itself if we choose the number of layers to be small. However, as we increase the number of QAOA layers, we see that it is more beneficial to perform ITE steps than to append QAOA layers. Specifically, QAOA with ITE can reach approximation ratios very close to one, compared to vanilla QAOA. 

\subsection{Maximum Independent Set}

The second classical optimization problem that we choose to tackle is the \emph{Maximum Independent Set} (MIS) problem. In this problem, the user is presented with a graph $G=(V,E)$, where $V$ is the set of vertices and $E$ is the set of edges of the graph, and the goal is to find the largest independent set, i.e., a subset of vertices $V'\subseteq V$ where no two vertices are connected by an edge. The MIS problem is stated as follows:
\begin{gather}
    \max_{\boldsymbol{x}\in\{0,1\}^n} \sum_{i\in V}x_i \\
    \text{subject to:} \; x_i+x_j<1,\; \forall(i,j)\in E
\end{gather}
where $x_i=1$ indicates whether the vertex $i$ is included in the independent set. Note that the constraint forces two connected vertices not to be both at the same time in the independent set. The Hamiltonian corresponding to the MIS can be expressed as \cite{vcepaite2025quantum}:
\begin{equation}
    H_{\text{MIS}} = \sum_{i\in V}\Big(\frac{1}{2} - \frac{\lambda d_i}{4}\Big)Z_i + \frac{\lambda}{4}\sum_{(i,j)\in E}Z_iZ_j
\end{equation}
where $d_i$ is the degree of the vertex $i$ and $\lambda$ is a Lagrange multiplier.

In Fig. \ref{fig:mis_performance}, we illustrate the performance of PDBQITE on MIS for Erdős–Rényi graphs when initialized with a QAOA generated state. To generate the QAOA state, we use the Graph Neural Network (GNN) approach outlined in \cite{vcepaite2025quantum}. We use both $p=3$ and $p=4$ QAOA layers and apply 5 PDBQITE iterations with a constant step-size. For our experiments, we use both 12 and 14 qubit instances and plot their average performance.

As we can see, PDBQITE constantly improves the approximation ratio with the number of PDBQITE iterations. In all experiments, after 5 iterations, the states that are generated by PDBQITE are very close to the true ground state, with approximation ratios very close to 1.

\begin{figure}[htbp]
    \centering
    \begin{subfigure}[b]{0.48\textwidth}
    \centering
    \caption{$p=3$}
    \label{fig:left}
    \includegraphics[width=\textwidth]{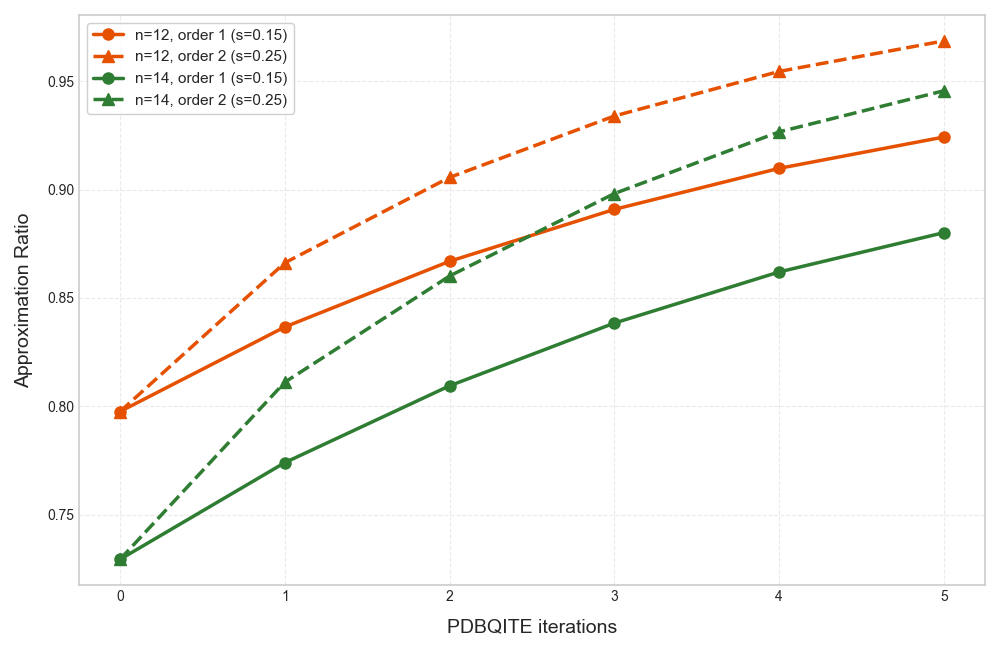}
    \end{subfigure}
    \hfill
    \begin{subfigure}[b]{0.48\textwidth}
        \centering
        \caption{$p=4$}
        \label{fig:right}
        \includegraphics[width=\textwidth]{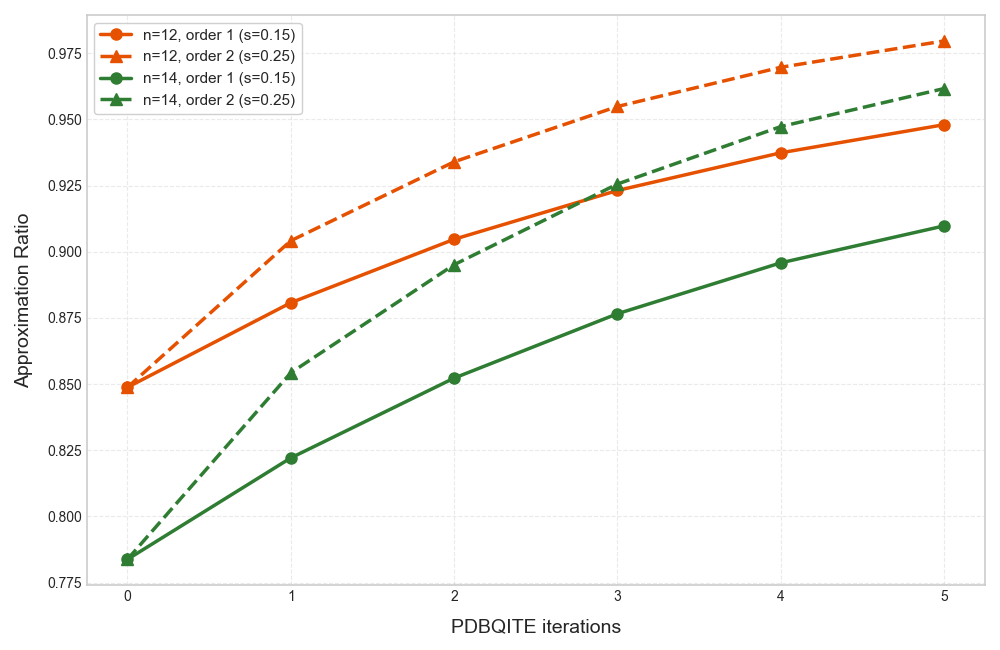} 
    \end{subfigure}

    \caption{Performance of first- and second-order PDBQITE on the MWIS problem with a QAOA-generated initial state. Within a few number of iterations, both methods can achieve approximation ratios larger than $r=0.9$.} 
    \label{fig:mis_performance}
\end{figure}

\subsection{$\mathrm{BeH_2}$}

The second molecule that we choose to prepare its ground state is $\mathrm{BeH_2}$. Here, instead of using a constant timestep, we choose to use a varying timestep with the number of iterations. In practice, we saw that at the beginning of the algorithm we can choose large timesteps $s_k$, while during the last steps of the algorithm, the timesteps should be reduced. Here again we see how second-order PDBQITE is able to prepare a quantum state that is within chemical accuracy in fewer iterations than first-order PDBQITE (and with higher probability of success). However, both methods are able to approximate the true ground state, indicating the usefulness of the method.

\begin{figure}[t]
    \centering
    \includegraphics[width=1\linewidth]{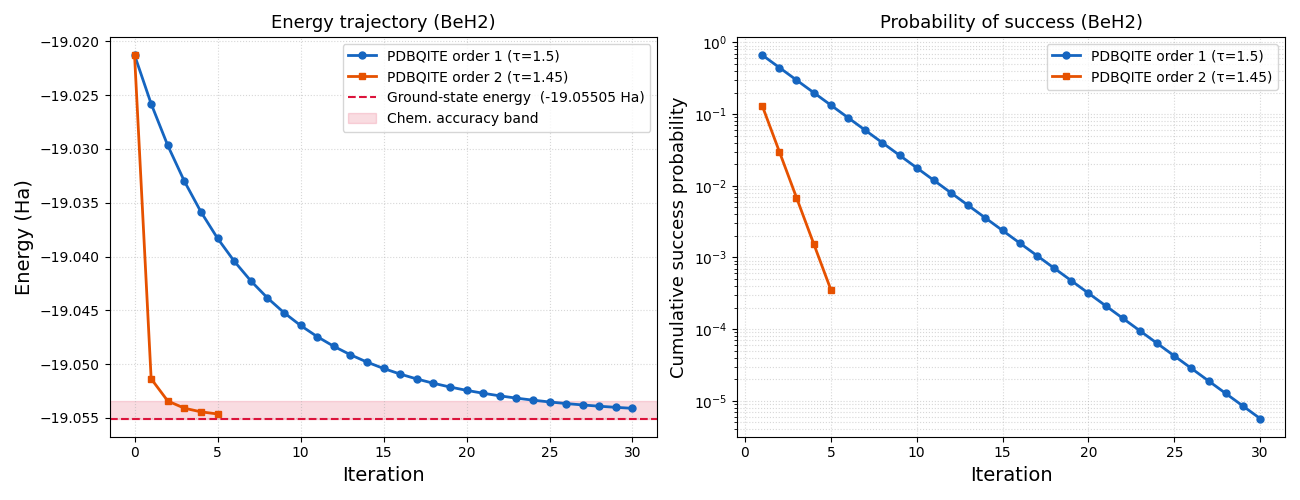}
    \caption{Energy trajectory (left figure) and cumulative probability of success (right figure) for the $\mathrm{BeH_2}$ molecule. Both first- and second-order PDBQITE can reach quantum states that are within chemical accuracy to the exact ground state.}
    \label{fig:beh2_molecule}
\end{figure}

\subsection{Heisenberg Model}

\begin{figure}
    \centering
    \includegraphics[width=1\linewidth]{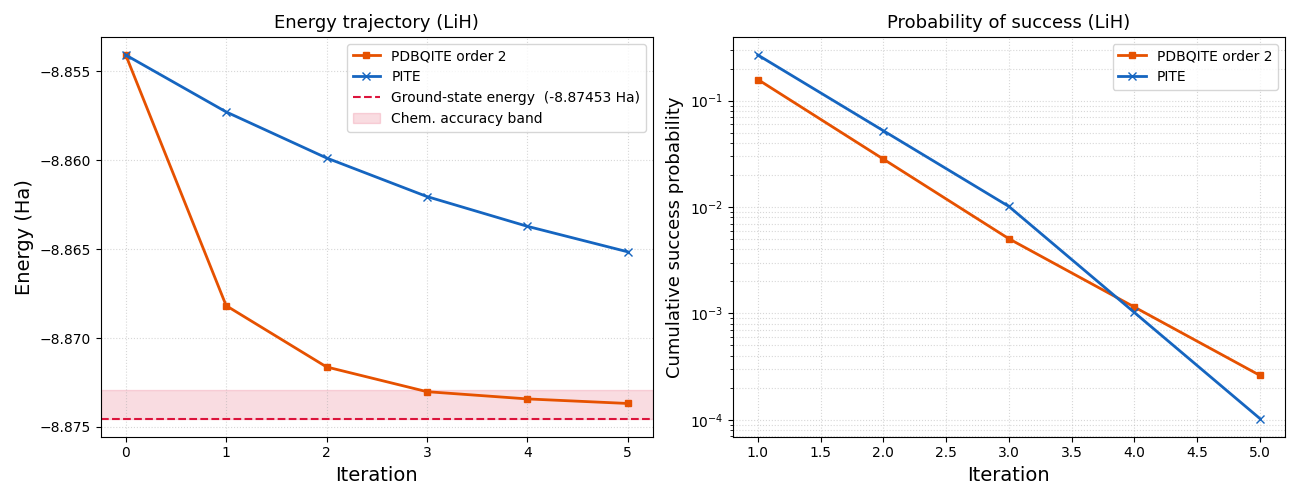}
    \caption{Comparison of second-order PDBQITE with PITE on the task of preparing the ground state of $LiH$ using the same timesteps for both bethod and $m_0=0.1$ for the PITE method.}
    \label{fig:pdbqite_pite}
\end{figure}

We also compared PDBQITE with PITE on the task of preparing the ground state of the antiferromagnetic Heisenberg model:
\begin{equation}
    H = \sum_{i=1}^{N-1} X_i X_{i+1} + Y_iY_{i+1} + Z_iZ_{i+1}
\end{equation}
where $N$ is the number of qubits. In this case, we used first-order PDBQITE, which is computationally cheaper than PITE, and performed experiments on $n=14$ qubits. The results are illustrated in Fig. \ref{fig:pdbqite_pite_heisenberg}. Overall, we can see that using less resources, first-order PDBQITE is able to reach states with energy similar to that of PITE but with a probability of success that is significantly higher than that of PITE.

\begin{figure}
    \centering
    \includegraphics[width=1\linewidth]{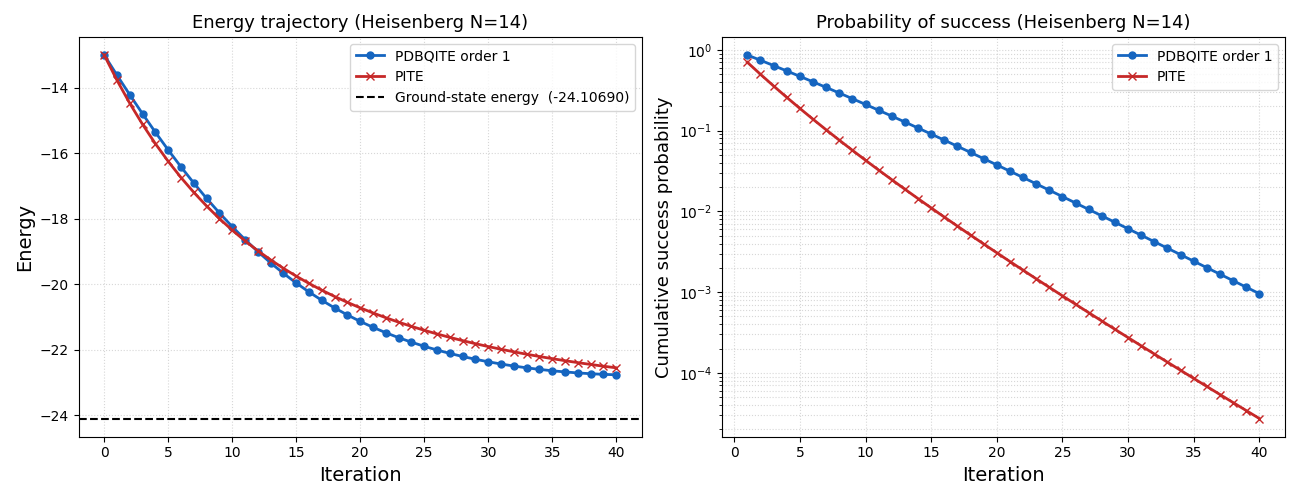}
    \caption{Comparison of first-order PDBQITE with PITE on the task of preparing the ground state of the antiferromagnetic Heisenberg model.}
    \label{fig:pdbqite_pite_heisenberg}
\end{figure}

\end{document}